%% file: CompressiveSensing-LectureNotes.tex
\documentclass{svproc}
\usepackage{url}
\usepackage{xcolor}
\newcommand{\corr}[1]{#1}
\usepackage{amsmath,amssymb}
\usepackage{tikz}
\usepackage{subcaption}
\usepackage{algpseudocode, algorithm, algorithmicx}

\input{Shortcuts}

\begin{document}
\mainmatter              
\title{Compressive Sensing -- Introduction and Relations to Deep Learning}
\titlerunning{Compressive Sensing}  
%
\author{Hung-Hsu Chou\inst{1}, Johannes Maly\inst{2,3}, Holger Rauhut\inst{2,3}}
%
\authorrunning{Chou, Maly, Rauhut} 
%
%

\institute{
University of Pittsburgh, Department of Mathematics, USA,
\and
Ludwig-Maximilians-Universität München, Department of Mathematics, Germany,
\and
Munich Center for Machine Learning, Germany}




\maketitle              

\begin{abstract}
Compressive sensing predicts that sparse vectors (signals) can be recovered from a \corr{small number} of linear measurements via efficient algorithms. This finding, \corr{which dates back two decades,} has triggered a paradigm shift in signal processing and initiated many developments both \corr{in practical signal-processing applications, such as medical imaging, radar, and astronomy,} and on the theoretical side.
More recently, seminal connections to the field of deep learning have \corr{led} to further advances \corr{in} the field, such as the use of unrolled neural networks for sparse recovery and the discovery that common training algorithms (variants of gradient descent) favor sparsity in overparameterized scenarios -- the so-called implicit bias phenomenon. This article gives an introduction to compressive sensing and outlines connections to deep learning. In particular, we will discuss generalization for neural networks generated by unrolling sparse recovery algorithms and implicit regularization for gradient descent applied to learning simplified linear neural networks.
\end{abstract}
\section{Introduction}
\label{sec:Introduction}

The field of compressive sensing was initiated
with seminal articles by Donoho \cite{donoho2006} and Cand{\'e}s, Romberg and Tao \cite{carota06}. While the traditional pathway of signal processing (before the advent of compressive sensing) was to first acquire (measure) a signal (or image) fully and then compress it in order to store it efficiently, the basic idea of compressive sensing is to try to directly acquire a compressed version of the signal using potentially much fewer measurements. In other words, the goal is to exploit the compressibility of a signal in order to be able to measure it more efficiently.
The discovery that this is actually possible has \corr{led} to an explosion of research \corr{activity} in the 2000s and 2010s, \corr{concerning} both practical and theoretical aspects. \corr{Methods} of compressive sensing are \corr{now} used effectively in various applications such as magnetic resonance imaging \cite{lustig2007sparse,hammernik2018}, radar \cite{herman2009high,ender2010compressive,Baidoo2021,dora17}, astronomy \cite{Wiaux2009,starck2010sparse,carillo2014}\corr{,} and many more. The mathematical analysis of compressive sensing methods is challenging and has \corr{led} to significant new developments in a number of mathematical fields\corr{,} including optimization and high-dimensional probability, in particular, nonasymptotic random matrix theory\corr{;} see e.g.\ \cite{foucart2013mathematical} for an overview and introduction.

The main ingredients for compressive sensing are \begin{itemize}
    \item \textbf{sparsity}: the signals to be recovered need to be sparse (within a suitable basis) or obey some other form of low complexity, for instance, \corr{low-rank structure};
    \item \textbf{efficient reconstruction}: fast reconstruction algorithms can be used for recovering a sparse signal from the measurement; most prominently $\ell_1$-minimization;
    \item \textbf{randomness}: provably effective measurement schemes are random.
\end{itemize}
The perhaps surprising ingredient is \corr{the} randomness of the measurements. While this is not necessary per se, so far all measurement schemes for which recovery could be shown with a minimal number of measurements are random. For deterministic measurement constructions, the proof techniques are not as powerful as for random constructions, and \corr{proving} optimal guarantees for deterministic constructions \corr{remains an open problem}.

The compressive sensing problem can be mathematically formulated as solving an underdetermined linear system for a sparse solution, see also the next chapter. The popular reconstruction method of $\ell_1$-minimization searches for the solution of the linear system of minimal $\ell_1$-norm.

Compressive sensing is potentially useful in applications if it is expensive, difficult, \corr{time-consuming,} or even impossible to take a lot of measurements and if it can be expected that the signals to be recovered are sparse in a suitable sense. This is often (but not always!) the case in signal processing tasks.

In recent years deep learning methods have entered the field of signal and image processing. A number of approaches combine these with compressive sensing methods, or rather enhance compressive sensing methods with machine learning techniques. 
Unrolling methods \cite{gregor2010learning,behboodi2022compressive} consider iterative algorithms for variational problems aiming at sparse recovery and interpret each iteration as one layer of a neural network. This point of view allows \corr{sparse reconstruction to be incorporated} into a machine learning pipeline where parameters of the iterative algorithm (such as the sparsity basis) can be learned using training data. Such approaches often considerably improve compressive sensing methods in practical applications.

Another direction where deep learning and compressive sensing meet on a rather theoretical level is the mathematical analysis of the so-called implicit bias phenomenon in deep learning. Modern machine learning models are often overparameterized meaning that the number of parameters exceeds the number of training samples by far. In this situation there are usually many neural networks interpolating the data exactly -- leading to zero training loss. Nevertheless, in practice learned networks often show good generalization to new data \cite{zhang2017understanding}, which is in contrast to classical statistics that would predict overfitting. In the situation of many global minimizer the employed optimization algorithm and its hyperparameters such as initialization have significant influence on the computed solution. Understanding such implicit regularization (also called implicit bias) of the training algorithm -- usually variants of (stochastic) gradient descent -- is key to the understanding of modern machine learning methods \cite{neyshabur2015,gunasekar_characterizing_2020}. While many fundamental questions on implicit bias are open, first theoretical insights could be achieved on simplified models, in particular, on gradient flow and gradient descent applied to linear neural networks \cite{gunasekar2017implicit,arora2019implicit,chou2023more,chou2024gradient,even2023s,pesme2023saddle}
as well as on two-layer networks with nonlinear activation function (such as ReLU), see e.g.\ \cite{boursier2024} and references therein.
One of these simplified models starts with a linear underdetermined inverse problem and replaces the unknown vector by a Hadamard product of several vectors, which can be viewed as a diagonal linear network. The corresponding \corr{squared} loss is then minimized via gradient flow (GF) -- an abstraction of gradient descent. It can be shown that\corr{,} for small initialization\corr{,} the limit of \corr{GF} is close to the $\ell_1$-minimizer of the \corr{linear} system. This reveals a close connection to compressive sensing, and shows\corr{,} in particular, that the implicit bias of gradient flow applied to linear networks is towards sparse solutions (for small initialization).

This observation leads to the conjecture that \corr{gradient-descent-type algorithms} favor \corr{solutions} of low complexity (such as sparsity) also for more realistic (nonlinear) neural network models. \corr{Occam's razor, the principle} that among many possible solutions nature favors the simplest ground-truth model\corr{,} suggests that a bias towards \corr{simple or sparse solutions} leads to good generalization in practice. Therefore, an implicit regularization of training algorithms towards low-complexity (sparse) solutions gives an intuitive explanation of why modern machine learning systems often generalize well. Of course, understanding this in detail still requires \corr{considerably more research}.

\medskip

This article is structured as follows. Section~\ref{sec:Basics-CS} provides an introduction to the basics of compressive sensing. It covers the notion of sparsity, recovery algorithms\corr{,} and random measurement matrices. Section~\ref{sec:Unrolled} is concerned with an iterative algorithm for sparse recovery, its interpretation as \corr{an} unrolled neural network\corr{,} and corresponding generalization bounds for trained unrolled networks. Section~\ref{sec:implicit-reg} introduces \corr{implicit regularization} in deep learning, covers \corr{the} theory \corr{of} training linear diagonal networks via gradient flow\corr{,} and \corr{discusses} the connections to compressive sensing.

\subsection{Notation}
\label{sec:notation}

Let us introduce some notation that is used throughout this article. For $0<p<\infty$, the $\ell_p$-norm of a vector $\bx \in \R^d$ is defined as $\|\bx\|_p = (\sum_{j=1}^d |x_j|^p)^{1/p}$ (note that $\|\cdot\|_p$ is only a quasi-norm for $0<p<1$). The operator norm of a matrix $\bA\in\mathbb{R}^{m\times d}$ from $\ell_p$ to $\ell_q$ is defined as $\|\bA\|_{p \to q} = \sup_{\|\bx\|_p = 1} \|\bA\|_q$. For a matrix $\bA \in \R^{m \times}$ with singular values $\sigma_{1} \geq \sigma_{2} \geq \sigma_{\min\{m,n\}}$, the Schatten $S_p$-norm is defined as $\|\bA\|_{S_p} = \|(\sigma_j)_{j=1}^{\min\{m,n\}}\|_p$. The special case $p=1$ gives the nuclear norm $\|\bA\|_* = \|(\sigma_j)\|_1 = \sum_{j} \sigma_j$. Moreover, $\|\bA\|_{S_\infty} = \|\bA\|_{2 \to 2}$ and $\|\bA\|_{S_2} = \|\bA\|_F = \sqrt{\operatorname{tr}(\bA^T\bA)}$ is the Frobenius norm (where $\operatorname{tr}(\bA)$ denotes the trace of $\bA$). The trace (or Frobenius) inner product is defined as $\langle \bA,\bB\rangle = \operatorname{tr}(\bA \bB^T)$ for matrices $\bA$, $\bB$ of matching dimensions.
For a symmetric matrix $\bA$ we write $\bA \succeq 0$ if it is positive semidefinite.
 
For $d \in \N$ we define $[d]:=\{1,2,\hdots, d\}$. For a vector $\bx \in \R^d$ and a subset $S \subset [d]$, the vector $\bx_S \in \R^d$ is the vector $\bx$ where the entries outside the set $S$ are set to zero, i.e.,
$(\bx_S)_j = \bx_j$ for $j \in S$ and $\bx_j = 0$ for $j \in S^c = [d] \setminus S$. For a set $T$, $|T|$ denotes its cardinality. The notation $a \lesssim b$ means $a \le Cb$, for an absolute constant $C > 0$ and $a \simeq b$ means $a \lesssim b$ and $b \lesssim a$. 
The probability of an event $A$ is denoted by $\Pr(A)$. The Gaussian distribution with mean $\mu$ and variance $\sigma^2$ is denoted $\calN(\mu,\sigma^2)$.

\section{Basics of Compressive Sensing}
\label{sec:Basics-CS}
 
Many practical problems in science can be well described in terms of linear inverse problems. In this formulation, a quantity of interest $\bx_\star \in \R^d$ has to be recovered from $m$ linear observations $\by \in \R^m$ of the form
\begin{align}
    \label{eq:CS}
    \by = \bA\bx_\star,
\end{align}
where $\bA \in \R^{m\times d}$ models the measurement process.\footnote{The content of this section can be generalized to complex-valued vector spaces in a straightforward way. We restrict ourselves here to $\R^d$ for the sake of simplicity.} Due to physical and economical constraints, one often encounters the situation $m \ll d$, i.e., the number of measurements is smaller than the signal length. This means that recovery of $\bx_\star$ from $\by$ is an ill-posed problem. At the same time, most real-world signals exhibit forms of intrinsic structure that is of lower complexity than the mere signal length suggests and that
may be leveraged in the reconstruction. One way to formalize this is the concept of sparsity.

\begin{definition}
\label{def:Sparsity}
The {\em support} of a vector  $\bx \in \R^d$ is 
the index set of its nonzero entries, i.e.,
$$
\supp(\bx) := \{ j \in [d]: x_j \not= 0 \}.
$$
The vector $\bx \in \R^d$ is called {\em $s$-sparse} if at most $s$ of its entries are nonzero,
i.e., if \label{eq:L0}
$$
\|\bx\|_0 := |\supp(\bx)| \le s.
$$ 
We denote the set of $s$-sparse vectors in $\R^d$ by $\Sigma_s^d$.
\end{definition}

The following lemma shows that, under mild assumptions on $\bA$, \eqref{eq:CS} can still be solved when $m < d$ provided that $\bx_\star$ is known to be $s$-sparse a priori.

\begin{lemma}[{\cite[Lemma 3.1]{cohen2009compressed}}] \label{lem:MinimalMeasurements}
	Given $\bA \in \R^{m\times d}$, the following properties are equivalent:
	\begin{enumerate}
		\item[(i)] Every $s$-sparse vector $\bx \in \R^d$ is the unique $s$-sparse solution of $\bA\bz = \bA\bx$, that is, if $\bA\bx = \bA\bz$ and both $\bx$ and $\bz$ are $s$-sparse, then $\bx = \bz$.
		\item[(ii)] The nullspace $\ker(\bA)$ does not contain any $2s$-sparse vector other than the zero vector, i.e., $\ker(\bA) \cap \Sigma_{2s}^d = \{ \0 \}$
		\item[(iii)] Every set of $2s$ columns of $\bA$ is linearly independent.
	\end{enumerate}
\end{lemma}

\begin{proof}
    $(i) \Rightarrow (ii)$: Assume $(i)$ holds and $\bz \in \ker(\bA) \cap \Sigma_{2s}^d$. Then $\bz$ can be written as $\bz = \bz_1 - \bz_2$ where $\bz_1,\bz_2 \in \Sigma_s^d$ and $\supp(\bz_1) \cap \supp(\bz_2) = \emptyset$. As $\bA\bz_1 - \bA\bz_2 = \bA\bz = \0$, $(i)$ implies $\bz_1 = \bz_2$. But as $\bz_1$ and $\bz_2$ have disjoint supports we obtain $\bz_1 = \bz_2 = \bz = \0$.\\
    $(ii) \Rightarrow (iii)$: Assume $(ii)$ and that $\bz \in \R^d$ encodes a linear combination of $2s$ columns of $\bA$ which yields zero, i.e., $\bz \in \Sigma_{2s}^d$ with $\bA\bz = \0$. Property $(ii)$ implies $\bz = \0$.\\
    $(iii) \Rightarrow (i)$: Assume $(iii)$ and that $\bx,\bz \in \Sigma_s^d$ yield the same measurements $\bA\bx = \bA\bz$. As $\bx - \bz \in \Sigma_{2s}^d$ and $\bA(\bx - \bz) = \0$, property $(iii)$ implies that $\bx = \bz$.
\end{proof}

According to Lemma \ref{lem:MinimalMeasurements} (iii) at least $m=2s$ observations are necessary to allow the unique identification of arbitrary $s$-sparse signals. It turns out that $m=2s$ is also sufficient to construct a suitable $\bA \in \R^{m\times d}$, for any $d \ge 2s$, see \cite[Theorem 2.14]{foucart2013mathematical}. If $\bA$ in \eqref{eq:CS} satisfies the condition in Lemma \ref{lem:MinimalMeasurements} (iii), $\bx_\star$ can thus be recovered by solving
\begin{align}
    \label{eq:L0min}
    \bx_\star = \argmin_{\bz \in \R^d} \| \bz \|_0
    \qquad \text{s.t. } \quad \bA\bz = \by.
\end{align}
This approach has two major shortcomings though. First of all, the program in \eqref{eq:L0min} is NP-hard to solve in general \cite{natarajan1995sparse}. (Intuitively, it requires to solve the restricted linear system for all $\binom{d}{s}$ 
possible support choices.) Second, only assuming the condition in Lemma \ref{lem:MinimalMeasurements} (iii) will not guarantee stable and robust recovery of $\bx_\star$, i.e., consistency of recovery in light of sparsity and measurement defects. Let us make this more precise. We call a vector $\bx$ \emph{compressible} if its error of best $s$-term approximation decays quickly in $s$.

\begin{definition}\label{def:errorbeststerm}
For $p>0$, the {\em $\ell_p$-error of best $s$-term approximation} to a vector $\bx \in \R^d$ is defined by
\label{NotBApprox}
$$
\sigma_s(\bx)_p := \inf
\big\{
\|\bx-\bz\|_p, \; \bz \in \R^d \mbox{ is $s$-sparse}
\big\}.
$$
\end{definition}

In practical applications, we do not expect that $\bx_\star$ is perfectly $s$-sparse, but rather that $\sigma_s(\bx_\star)_p$ is small, for $s \ll d$ and suitably chosen $p$. Furthermore, we hardly ever encounter perfect observations as in \eqref{eq:CS}. A more realistic model is given by
\begin{align}
    \label{eq:CSnoise}
    \by = \bA\bx_\star + \Eta,
\end{align}
where $\Eta \in \R^m$ models unknown noise. A reliable reconstruction method $\calA$ should satisfy that the reconstruction error $\| \calA(\by,\bA) - \bx_\star \|_p$ is proportional to the sparsity defect as measured by $\sigma_s(\bx_\star)_p$ (stability) and the observation noise as measured by $\| \Eta \|_2$ (robustness). 

A central insight of compressive sensing is that the above shortcomings can be mended by relaxing the program in \eqref{eq:L0min} to a convex program (leading to tractable algorithms) and posing stronger assumptions on $\bA$.

\begin{figure}[!t]
    	\centering
    	\begin{tikzpicture}
            \coordinate[label=right:{\large $\bA\bz = \bA\bx$}] (B) at (1,2.2);
            
            \draw[very thick,->] (-3,0) -- (3,0);
            \draw[very thick,->] (0,-3) -- (0,3);
            
            \draw[dashed,thick] (2,0) to[out=180,in=90] (0,-2) to[out=90,in=0] (-2,0) to[out=0,in=-90] (0,2) to[out=-90,in=180] cycle;
            \draw[thick] (2,0) -- (0,-2) -- (-2,0) -- (0,2) -- cycle;
            \draw[thick,dotted] (0,0) circle (52pt);

            \filldraw (0,2) circle (2pt);
            \filldraw (0.6,1.75) circle (2pt);
            
            \draw[thick] (-2.5,3) -- (2.5,1);
            
        \end{tikzpicture}
    	\caption{Interplay between $\ell_p$-balls and sparsity of the respective minimizer \cite{maly2019recovery}.}
    	\label{fig:BP}
\end{figure}
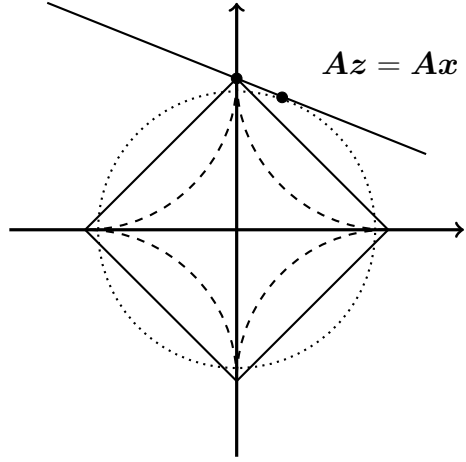

\subsection{Basis pursuit and the null space property}

When revisiting the definition of $\|\bx\|_0$, one notices that, although it is abusively called $\ell_0$-norm, it is neither a norm nor a quasinorm. The convention of using the norm notation 
comes from the observation that 
\begin{align}
    \label{eq:lpnormLimit}
\|\bx\|_p^p := \sum_{j=1}^d |x_j|^p \underset{p \to 0}{\longrightarrow} 
\sum_{j=1}^d {\bf 1}_{ \{x_j \not= 0 \}}
= | \{ j \in [\N]: x_j \not= 0 \} |,
\end{align}
where we define ${\bf 1}_{ \{x_j \not= 0 \}} = 1$ if $x_j \neq 0$, and ${\bf 1}_{ \{x_j \not= 0 \}} = 0$ if $x_j = 0$.
In other words the quantity $\|\bx\|_0$
is the limit of the $p$th power of the $\ell_p$-quasinorm of $\bx$ as $p$ decreases to zero. Due to \eqref{eq:lpnormLimit} we can view $\| \bx \|_p^p$, for $p > 0$, as a natural relaxation of $\| \bx \|_0$. For $p \in (0,1)$, this leads to a non-convex function which is still difficult to optimize. By choosing $p=1$, we obtain a convex relaxation of \eqref{eq:L0min}, 
\begin{align}
\label{eq:BP}
    \min_{\bz \in \R^d} \| \bz \|_1
    \qquad \text{s.t. } \quad \bA\bz = \by,
\end{align}
which is often called \emph{basis pursuit} and can be solved algorithmically in polynomial time. 

Figure \ref{fig:BP} illustrates why minimizing $\ell_p$-norms, for $p \le 1$, yields sparse solutions in contrast to other norms such as $\ell_2$. The figure further suggests that the validity of \eqref{eq:BP} as relaxation of \eqref{eq:L0min} depends on the kernel geometry of $\bA$, since alignment of the solution space $\{\bA\bz = \by\}$ with the facets of the $\ell_1$-ball would lead to non-sparse minimizers. This is confirmed and formalized by the so-called null space property (NSP), \cite{doel03,grni07,codade09}, \cite[Definition 4.1]{foucart2013mathematical}.

\begin{definition}[Null space property] \label{def:NSP}
    A matrix $\bA \in \R^{m\times d}$ is said to satisfy the \emph{null space property} relative to a set $S \subset [d]$ if
	\begin{align*}
		\pnorm{\bz_S}{1} < \pnorm{\bz_{S^c}}{1} \quad \text{for all } \bz \in \ker(\bA) \setminus \{ \0 \}.
	\end{align*}
	It is said to satisfy the null space property of order $s$ if it satisfies the null space
    property relative to any set $S \subset [d]$ with $|S| \le s$.
\end{definition}

\begin{remark}
    The NSP of order $s$ implies property $(ii)$ in Lemma \ref{lem:MinimalMeasurements}.
\end{remark}

It turns out that $\bA$ satisfies the NSP if and only if \eqref{eq:L0min} and \eqref{eq:BP} are equivalent, i.e., the NSP fully characterizes when \eqref{eq:L0min} can be solved by its convex relaxation \cite{codade09,grni07},  \cite[Theorem 4.5]{foucart2013mathematical}

\begin{theorem}
\label{thm:NSP}
	Given a matrix $\bA \in \R^{m\times d}$ , every $s$-sparse vector $\bx \in \R^d$ is the unique minimizer of \eqref{eq:BP} if and only if $\bA$ satisfies the null space property of order $s$.
\end{theorem}

\begin{proof}
    Consider first one fixed support set $S \subset [d]$ with $|S| \le s$. Assume that every $\bx \in \R^d$ with $\supp(\bx) \subset S$ is the unique minimizer of \eqref{eq:BP}. Hence, for any $\bv \in \ker(\bA) \setminus \{ \0 \}$, the vector $\bv_S$ is the unique minimizer of $\pnorm{\bz}{1}$ subject to $\bA\bz = \bA\bv_S$. This implies $\pnorm{\bv_S}{1} < \pnorm{\bv_{S^c}}{1}$ as $-\bv_{S^c} \neq \bv_S$ and by $\bA(\bv_S + \bv_{S^c}) = \bA\bv = \0$ one has $\bA(-\bv_{S^c}) = \bA\bv_S$.
    
    Conversely, let us assume that the NSP relative to $S$ holds. Given $\bx \in \R^d$ with $\supp(\bx) \subset S$ and $\bz \in \R^d$ such that $\bx \neq \bz$ and $\bA\bx = \bA\bz$, define $\bv = \bx - \bz \in \ker(\bA)$. We get that
    \begin{align*}
        \pnorm{\bx}{1} &\le \pnorm{\bx - \bz_S}{1} + \pnorm{\bz_S}{1} = \pnorm{\bv_S}{1} + \pnorm{\bz_S}{1} \\
        &< \pnorm{\bv_{S^c}}{1} + \pnorm{\bz_S}{1} = \pnorm{\bz}{1},
    \end{align*}
    which shows that $\bx$ is the unique minimizer of \eqref{eq:BP}. The claim follows by varying over all possible support sets.
\end{proof}

Theorem \ref{thm:NSP} shows that the NP-hardness of \eqref{eq:L0min} can be circumvented. To guarantee stable and robust recovery of $\bx_\star$ under the noisy observation model \eqref{eq:CSnoise}, the NSP has to be strengthened \cite{grni07,codade09}, \cite[Definition 4.17]{foucart2013mathematical}.

\begin{definition}[{Stable and robust NSP}] \label{def:srNSP}
    A matrix $\bA \in \R^{m\times d}$ is said to satisfy the \emph{stable and robust null space property} with constants $0 < \rho < 1$ and $\tau > 0$ relative to a set $S \subset [d]$ if
	\begin{align*}
		\pnorm{\bz_S}{1} < \rho \pnorm{\bz_{S^c}}{1} + \tau \pnorm{\bA\bz}{2} \quad \text{for all } \bz \in \R^d.
	\end{align*}
	It is said to satisfy the stable and robust null space property of order $s$ if it satisfies the stable and robust null space property with constants $0 < \rho < 1$ and $\tau > 0$ relative to any set $S \subset [d]$ with $|S| \le s$.
\end{definition}

To recover $\bx_\star$ under the noisy observation model in \eqref{eq:CSnoise}, one usually relaxes the constraints in \eqref{eq:BP}. For $\eta \ge 0$, consider
\begin{align} \label{eq:BPD}
    \min_{\bz \in \R^N} \pnorm{\bz}{1}, \quad \text{subject to } \pnorm{\bA\bz - \by}{2} \le \eta,
\end{align}
which is commonly known as \emph{quadratically constrained basis pursuit}. As the following theorem shows, it allows stable and robust recovery in polynomial time if $\eta$ is suitably chosen in dependence on the noise level $\pnorm{\Eta}{2}$, see also \cite{codade09}, \cite[Theorem 4.19]{foucart2013mathematical}.

\begin{theorem}\label{thm:srNSP}
    Suppose that $\bA \in \R^{m\times d}$ satisfies the stable and robust NSP of order $s$ with constants $0 < \rho < 1$ and $\tau > 0$. Then, for any $\bx \in \R^d$ with measurements \eqref{eq:CSnoise}, a solution $\hat{\bx}$ of \eqref{eq:BPD} with $\eta \ge \pnorm{\Eta}{2}$ fulfills
    \begin{align} \label{eq:srNSP}
        \pnorm{\bx - \hat{\bx}}{1} \le \frac{2(1+\rho)}{1-\rho} \sigma_s(\bx)_1 + \frac{4\tau}{1-\rho} \eta. 
    \end{align}
\end{theorem}

There remain three fundamental questions. How many observations are required to have the (stable and robust) NSP? Which types of matrices $\bA$ achieve these bounds? And, finally, how do we solve \eqref{eq:BPD} in practice? In the following two sections, we will provide concise answers.

\subsection{Random Measurement Matrices}

Although null space properties allow a rigorous characterization of the solvability of \eqref{eq:BP} and \eqref{eq:BPD}, it is often convenient to work with a stronger recovery condition. The \emph{restricted isometry property} (RIP) of a matrix $\bA \in \R^{m\times d}$ has been introduced in \cite{candes2006near} under the name "uniform uncertainty principle" and implies the stable and robust NSP, see also \cite[Definition 6.1]{foucart2013mathematical}.

\begin{definition}[{Restricted Isometry Property}] \label{def:RIP}
	A matrix $\bA \in \R^{m\times d}$ satisfies the restricted isometry property of order $s$ ($s$-RIP) with constant $0 < \delta_s < 1$ if, for all $\bz \in \Sigma_s^d$,
	 \begin{align} \label{eq:RIP2}
	 (1-\delta_s) \Vert \bz \Vert_2^2 \le \Vert \bA\bz \Vert_2^2 \le (1+\delta_s) \Vert \bz \Vert_2^2.
	 \end{align}
\end{definition}


The RIP implies stable and robust recovery via $\ell_1$-minimization.
Several versions of such theorems have been shown over the years \cite{cata05,carota06-1,ca08,cawaxu10,foucart2013mathematical,cazh14}. 
The constants in these results have been improved subsequently and we state the result with the best available constant $1/\sqrt{2}$ \cite{cazh14}, which is actually optimal \cite{dagr09}.

\begin{theorem}
\label{thm:RIP}
    If $\bA \in \R^{m\times d}$ satisfies the $2s$-RIP with 
    $\delta_{2s} < 1/\sqrt{2} \approx 0.71$ then
    $\bA$ satisfies the stable and robust NSP with constants $\rho = c_\delta$ and $\tau = \sqrt{s} c'_\delta$ where $0 < c_\delta < 1$ and $c'_\delta > 0$ only depend on $\delta=\delta_{2s}$.
\end{theorem}

Whereas $\bA$ cannot fully preserve the geometry (distances and angles) of $\R^d$ if $m \ll d$, it acts almost like an isometry when restricted to sparse vectors and thus preserves the geometry of $\Sigma_s^d$. 

In contrast to the noiseless case \eqref{eq:CS} with $s$-sparse ground-truth, $m = 2s$ observations do not suffice in \eqref{eq:CSnoise} to guarantee stable reconstruction as in \eqref{eq:srNSP}. Viewing the problem through the lens of Gelfand widths of $\ell_p$-balls, one can show that
\begin{align} \label{eq:LowerBound}
    m \ge Cs \log \left( \frac{ed}{s} \right)
\end{align}
linear measurements are necessary to obtain uniform error bounds as in \eqref{eq:srNSP} \cite{cohen2009compressed}. Here, $C > 0$ denotes an absolute constant that is independent of $s$ and $d$. When comparing \eqref{eq:LowerBound} with Lemma \ref{lem:MinimalMeasurements}, we see that this worsens the necessary number of observations by a multiplicative factor which is logarithmic in $d$. So far there exists no deterministic method to construct RIP matrices achieving the bound in \eqref{eq:LowerBound}. The best known methods \cite{bodifokoku11} require $m \ge C s^{2-\varepsilon}$, for some very small $\varepsilon >0$, i.e., the number of observations scales (approximately) quadratically. 

In order to overcome this so-called quadratic bottleneck one passes to random measurement matrices. This allows to use very powerful and versatile tools from the field of high-dimensional probability, see e.g. \cite{vershynin2010introduction}. 
Gaussian random matrices represent a very common model, which can be analyzed in comparably simple ways in the context of compressive sensing.
An \corr{$m \times d$ matrix $\bA$} whose entries
$a_{i,j}$ are independent mean-zero and variance one Gaussian random variables ($a_{i,j} \sim \mathcal{N}(0,1)$) is called a Gaussian random matrix. The following result on the restricted isometry property of Gaussian random matrices has been proved in many variants and generalizations, see e.g.\ \cite{candes2006near,Baraniuk2008,mepato09,foucart2013mathematical}. It implies in particular that \eqref{eq:LowerBound} is also a sufficient estimate for the required number of measurements.

\begin{theorem}
\label{thm:RIP_Gaussian}
    Let $\bA \in \R^{m\times d}$ be a random draw of a Gaussian matrix and assume that for $\delta, \varepsilon \in (0,1)$,
    \begin{equation}\label{m:bound}
    m \geq C \delta^{-2} \left(s \log(\corr{ed}/s) + \log(1/\varepsilon)\right),
    \end{equation}
    where $C>0$ is some absolute constant. Then with probability at least $1-\varepsilon$ the matrix $\frac{1}{\sqrt{m}} \bA$ satisfies the $s$-RIP with constant $\delta_s \leq \delta$.
%
%
\end{theorem}

\begin{proof}[sketch] Several proofs of this theorem are known. The \corr{probably} simplest one
    consists of three main steps which form a prototypical argument. 
    First, it follows from the Bernstein inequality that the matrix $\frac{1}{\sqrt{m}} \bA$ fulfills
    \begin{align} \label{eq:RIPproof1}
        \Pr\left( \big| \pnorm{\bA\bz}{2}^2 - \pnorm{\bz}{2}^2 \big| \ge \varepsilon \pnorm{\bz}{2}^2 \right) \le 2e^{-cm\varepsilon^2},
    \end{align}
    for some absolute constant $c > 0$ and any fixed $\bz \in \R^d$, see e.g.\ \cite[Proposition 5.16]{vershynin2010introduction}.
    Then one constructs a $\frac{\delta}{4}$-net of $\Sigma_s^d \cap \mathbb S^{d-1}$ and applies a union bound to extend \eqref{eq:RIPproof1} to hold for all points in the net simultaneously with $\varepsilon = \frac{\delta}{2}$.\footnote{As \eqref{eq:RIP2} is invariant under $\ell_2$-norm scaling one can restrict the argument to $\bz \in \mathbb S^{d-1}$.}
    Finally, one uses stability of \eqref{eq:RIP2} against small perturbations to extend the result to $\Sigma_s^d \cap \mathbb S^{d-1}$ by paying a factor two in $\delta$. A concise version of this proof can be found in \cite{Baraniuk2008,foucart2013mathematical}.
\end{proof}

An explicit value of the constant $C>0$ in \eqref{m:bound} can be given, see e.g.\ \cite[Remark 9.28]{foucart2013mathematical}. The theorem generalizes to so-called subgaussian distributions for the entries $a_{i,j}$ including Bernoulli random variables \cite{mepato09,foucart2013mathematical}.

In combination, Theorems \ref{thm:srNSP}, \ref{thm:RIP}, and \ref{thm:RIP_Gaussian} show that robust recovery of approximately $s$-sparse signals via tractable optimization is possible if $m \gtrsim s \log(ed/s)$ and $\bA$ is suitably chosen. 

In most practical applications Gaussian matrices are undesirable because they are dense and unstructured and therefore imply high computational costs in storing and processing. This motivates the study of structured random matrices \cite{ra10}, where the structure is often motivated by physical (or other) constraints on the measurement process and the (limited) randomness allows for an (almost) optimal analysis of the required number of measurements.
In particular, Theorem~\ref{thm:RIP_Gaussian} can be extended to various types of structured random matrices at the cost of additional logarithmic factors in \eqref{eq:LowerBound}, see e.g.\ \cite{cata06,dora17,ra10,krmera14,ro09}.
For the purpose of this exposition, we limit ourselves on one type of structured random matrices, namely sampling matrices associated to bounded orthonormal systems.

Let $\calD \subset \R^n$ be endowed with a probability measure $\nu$ and assume that 
$\Phi = \{\phi_1,\hdots,\phi_d\}$ is an orthonormal system of complex-valued 
functions on $\calD$, that is, for $j,k \in [d]$,
\begin{equation}\label{def:orthonormality}
\int_{\calD} \phi_j(\bt) \overline{\phi_k(\bt)} d\nu(\bt) = \delta_{j,k} = \left\{ \begin{array}{cl} 0 & \mbox{ if } j \neq k,\\
1 & \mbox{ if } j= k.  
\end{array}
\right.
\end{equation}
We consider functions on $\calD$ that can be expanded into the orthonormal system $\Phi$ as
\[
f(\bt) = \sum_{j=1}^d x_j \Phi(\bt).
\]
The function $f$ is called $s$-sparse if $\|\bx\|_0 \leq s$. Given a sequence of points $\bt_1, \hdots, \bt_m \in \calD$ and corresponding samples of $f$,
\[
y_\ell = f(\bt_l) = \sum_{j=1}^d x_j \phi_j(\bt_\ell), 
\]
the task is to reconstruct $f$ (or equivalently $\bx$) from the vector $\by$ of samples. Introducing the sampling matrix 
$\bA \in \C^{m \times d}$ with entries
\begin{equation}\label{A:sample}
    A_{\ell,k} = 
    \phi_k(\bt_\ell), \qquad \ell \in [m], k \in [d]
\end{equation}
we have
\[
\by = \bA \bx
\]
so that the reconstruction of a sparse function $f$ from few samples amounts to solving a compressive sensing problem.
We introduce randomness in the structured matrix $\bA$ by choosing the sampling points $\bt_1,\hdots,\bt_m$ independently at random according to the probability measure $\nu$ on $\calD$. The resulting matrix $\bA \in \C^{m \times d}$ is then called random sampling matrix.

In order to show sparse recovery results, more precisely, bounds on the RIP of the random sampling matrix $\bA$\corr{, we require} a certain boundedness property of the orthonormal function system, which is stated next.

\begin{definition}[{\cite[Definition 12.1]{foucart2013mathematical}}]
    We call $\Phi =  \{\phi_1,\hdots,\phi_d\}$ a {\em bounded orthonormal system} (BOS)\index{bounded orthonormal system}\index{BOS} with constant $K$ 
    if it satisfies \eqref{def:orthonormality} and if
    \begin{equation}\label{def:KK}
    \|\phi_j \|_\infty := \sup_{\bt \in \calD} |\phi_j(\bt)| \leq K \quad \mbox{ for all } j \in [d] .
    \end{equation}
\end{definition}

The most prominent example of a BOS is the Fourier system (trigonometric polynomials), where $K=1$. By choosing $\calD$ to be an equidistant grid on $[0,1]$, this leads to randomly subsampled discrete Fourier matrices as associated random sampling matrices. 

It is important for the RIP bound that $K \lesssim 1$, or at least that $K$ grows only very mildly with $d$.
Many function systems including orthogonal polynomials (after applying some form of preconditioning) \cite{rawa10,rasc17}, spherical harmonics \cite{rawa16} and more satisfy this condition. We refer also to \cite[Chapter 12]{foucart2013mathematical} for more detailed examples.

Bounds on the RIP of random sampling matrices associated to a BOS  have been obtained in \cite{cata06,ru06-1,foucart2013mathematical,Bourgain2014RIP,Haviv15RIPFT,ChkifDextTranWebs16,BDJR21} (with successively improved logarithmic factors).
The currently best available bound \cite{Bourgain2014RIP,Haviv15RIPFT,ChkifDextTranWebs16,BDJR21} is presented below. 

\begin{theorem}
\label{thm:BOS:RIP:simple} 
    There is an absolute constant $C > 0$ such that the following holds. Let $\bA \in \C^{m \times d}$ be the 
    random sampling matrix  
    associated to a BOS 
    with constant $K \geq 1$. 
    If, for $\delta \in (0,1)$, 
    \begin{equation}
    \label{m:BOS:RIP:simple}
        m \geq C K^2 \delta^{-2} s \log^2(Ks/\delta) \log(ed),
    \end{equation}
    then with probability at least $1- d^{-\log^3(d)}$ the restricted isometry constant $\delta_s$ of $\frac{1}{\sqrt{m}} \bA$
    satisfies $\delta_s \leq \delta$.
\end{theorem}


Theorem~\ref{thm:BOS:RIP:simple} guarantees in particular that randomly subsampled Fourier matrices satisfy the RIP with high probability if (roughly speaking)
\begin{equation}\label{m:BOS}
m \gtrsim s \log^3(d).
\end{equation}
This bound is only slightly more restrictive than the one for Gaussian matrices. At the same time, discrete Fourier matrices are highly structured, require no storage space, and allow fast matrix-vector products via the fast Fourier transform (FFT). 

We note \corr{that} the bound \eqref{m:BOS} on the number of \corr{required} samples can be slightly improved in terms of the logarithmic factors (leading to the condition $m \gtrsim s \log(d)$) when considering so-called non-uniform recovery guarantees which hold for a fixed sparse vector and a random draw of the random sampling matrix -- rather than uniformly for all sparse \corr{vectors} for one draw of the random sampling matrix. We refer to \cite{carota06-1,foucart2013mathematical,ra10} for details.

\subsection{Recovery Algorithms}

Over the years, various algorithms have been proposed for the efficient reconstruction of $\bx_\star$ from \eqref{eq:CSnoise}. One class of methods such as \emph{orthogonal matching pursuit} (OMP) \cite{pati1993orthogonal} or \emph{compressive sampling matching
pursuit} (CoSaMP) \cite{needell2010cosamp} constructs a sparse solution $\hat\bx$ by iteratively shaping an initially empty support and minimizing a residual error. These are often referred to as \emph{greedy algorithms}, see also \cite{foucart2013mathematical}. Other types of methods aim to approximate solutions to the optimization problems in \eqref{eq:L0min} or \eqref{eq:BPD}. We will only discuss one of these latter methods in detail, the \emph{iterative soft-thresholding algorithm (ISTA)} \cite{dadede04}. 


ISTA alternates between performing gradient descent steps on the least-squares functional $\bx \mapsto \frac{1}{2} \| \bA\bx - \by \|_2^2$ and applying the \emph{soft-thresholding operator} $\St_\lambda \colon \R^d \to \R^d$ defined entry-wise via
\begin{align} \label{eq:St}
    (\St_\lambda (\bx))_i = \begin{cases}
    x_i - \lambda    & x_i > \lambda \\
    0               & |x_i| \le \lambda \\
    x_i + \lambda    & x_i < -\lambda,
    \end{cases}
\end{align}
for $\lambda > 0$. Starting with an initialization $\bx^0$ it performs the iterations
\begin{equation}
\label{eq:ISTA}
\bx^{l+1} = \mathbb S_{\lambda/2}(\bx^l -\bA^T(\bA\bx^l-\by)).
\end{equation}
ISTA is a special instance of proximal gradient descent \cite{parikh2014proximal} (or forward-backward splitting) applied to the optimization problem
\begin{align}
\label{eq:BPDN}
    \min_{\bz \in \R^d} \| \bA\bz - \by \|_2^2 + \lambda \| \bz \|_1.
\end{align}
Note that \eqref{eq:BPDN} goes under the name \emph{basis pursuit denoising (BPDN)}, and is equivalent to \eqref{eq:BPD} \cite[Proposition 3.2]{foucart2013mathematical} (but the corresponding relation between the parameters $\lambda$ and $\eta$ also depends on the optimal solutions). 
ISTA converges under mild conditions on $\bA$.

\begin{theorem}[{\cite{combettes2005signal}}] \label{thm:ISTA}
    Let $\bA \in \R^{m\times d}$ and $\by \in \R^m$. If $\pnorm{\bA}{2\rightarrow 2} < \sqrt{2}$, the sequence of iterates $\bx^l$ in 
    \eqref{eq:ISTA}
    converges to a minimizer of \eqref{eq:BPDN} as $\ell \to \infty$.
\end{theorem}

The assumption on $\pnorm{\bA}{2\rightarrow 2}$ is always fulfilled by a proper rescaling of \eqref{eq:BPDN}. If one replaces $\bA$, $\by$, and $\lambda$ by $\bA/\pnorm{\bA}{2\rightarrow 2}$, $\by/\pnorm{\bA}{2\rightarrow 2}$, and $\lambda/\pnorm{\bA}{2\rightarrow 2}^2$, the minimizers of \eqref{eq:BPDN} do not change but Theorem \ref{thm:ISTA} applies.
The convergence of ISTA can be rather slow in general, but there exist accelerated versions such as \emph{fast ISTA (FISTA)} \cite{beck2009fast}.

Due to the equivalence of \eqref{eq:BPDN} and \eqref{eq:BPD}, one can easily derive reconstruction guarantees for $\bx_{\text{ISTA}}$ computed by ISTA from Theorem~\ref{thm:srNSP}.

\begin{theorem} \label{thm:BPDN}
    Suppose that the $2s$-th RIP constant of $\bA \in \R^{m \times d}$ satisfies $\delta < 4/\sqrt{41} \approx 0.6$. Then, for any $\bx \in \R^d$ and $\by \in \R^m$ with $\|\bA\bx -\by\|_2 \leq \eta$ the following holds. Denote by $\bx_\lambda$ a minimizer of \eqref{eq:BPDN}. If $\lambda > 0$ is sufficiently large such that $\eta_\lambda \colon= \| \by - \bA\bx_\lambda \|_2 \ge \eta$, one has that
    \begin{align*}
        \|\bx - \bx_\lambda \|_1 \leq C \sigma_s(\bx)_1 + D\sqrt{s} \eta_\lambda,
    \end{align*}
    for constants $C$ and $D$ depending on the RIP constant $\delta$ of $\bA$.
\end{theorem}
\begin{proof}
    Theorem \ref{thm:BPDN} is a re-statement of Theorem \ref{thm:srNSP}. We just use Theorem \ref{thm:RIP} and the fact that the minimizer $\bx_\lambda$ is a solution to \eqref{eq:BPD} with $\eta$ replaced by $\eta_\lambda$ \cite[Proposition 3.2]{foucart2013mathematical}.
\end{proof}

\subsection{Low Rank Matrix Recovery}
\label{sec:LowRank}

The techniques of compressive sensing can be extended to the recovery of a matrix $\bX \in \R^{d_1 \times d_2}$ of low rank from linear measurements
\begin{equation}\label{matrix-measurements}
\by = \mathcal{A}(\bX) \in \R^m,
\end{equation}
where $\mathcal{A} : \corr{\R^{d_1 \times d_2}} \to \R^m$ is a linear map with $m < d_1 d_2$ that is usually described via matrices \corr{$\bA_1,\hdots,\bA_m$} as
\begin{equation}\label{matrix-meas}
y_j = \langle \bX, \bA_j \rangle_F = \operatorname{tr}(\bX^T \bA_j), \quad j=1,\hdots, m.
\end{equation}
Note that for a matrix $\bX \in \R^{d_1 \times d_2}$ with singular values $\sigma_1\geq \sigma_2 \geq \sigma_d \geq 0$, where $d = \min\{d_1,d_2\}$, the rank of $\bX$ equals the $\ell_0$-norm of the vector $\boldsymbol{\sigma} = (\sigma_j)_{j=1}^d$. 

A special case of this low rank matrix recovery problem is the matrix completion problem where the measurements are entries of $\bX$, see e.g.\ \cite{candes2009exact,recht2011simpler,gross2011recovering}.

For recovery of $\bX$ from $\by$, the naive approach of minimizing the rank subject to the linear constraint \eqref{matrix-measurements} is NP-hard. This actually follows from NP-hardness of $\ell_0$-minimization for sparse recovery.
Considering the success of the $\ell_1$-minimization approach for the latter, it suggests itself to consider 
the nuclear norm minimization approach
\[
\min_{\bZ} \|\bZ\|_* \quad \mbox{ subject to } \mathcal{A}(\bZ) = \by,
\]
where the nuclear norm is actually the $\ell_1$-norm of the vector of singular values, $\|\bX\|_* = \sum_{j} \sigma_j$. This is a convex optimization problem which can be solved via semidefinite programming and other tractable approaches.

The theory for \corr{low-rank} matrix recovery is developed analogously to the theory of sparse vector recovery. In particular, versions of the null space property and the restricted isometry property for the measurement map $\mathcal{A}$ have been introduced which imply exact recovery of matrices $\bX \in \R^{d_1 \times d_2}$ of rank at most $r$ from \corr{$\mathcal{A}(\bX)$}. For instance, $\mathcal{A}$ is said to have the restricted isometry property (RIP) of order $r$ with constant \corr{$\delta_r$} if
\begin{equation}\label{matrix-RIP}
(1-\delta_r) \|\bX\|_F^2 \leq \|\mathcal{A}(\bX)\|_2^2 \leq (1+\delta_r) \|\bX\|_F^2 \quad 
\forall \bX \in \R^{d_1 \times d_2} \mbox{ with } \operatorname{rank}(\corr{\bX}) \leq r,
\end{equation}
where $\|\cdot\|_F$ denotes the Frobenius norm. 
A random Gaussian measurement map $\mathcal{A} : \R^{d_1 \times d_2} \to \R^m$ is generated by $m$ independent standard Gaussian random matrices $\bA_i \in \R^{d_1 \times d_2}$ in \eqref{matrix-meas}. 
The map $\frac{1}{\sqrt{m}} \mathcal{A}$ satisfies the RIP of order $r$, i.e., $\delta_r \leq \delta \in (0,1)$ with probability at least $1-\epsilon$ if \cite{candes2011tight,recht2010guaranteed}
\begin{equation}\label{cond:m:RIP-matrix}
m \geq C \delta^{-2}(r(d_1+d_2) + \log(2/\varepsilon)),
\end{equation}
where $C>0$ is a universal constant.
This implies that $d_1 \times d_2$ matrices of rank $r$ can be recovered successfully from $m$ measurements via nuclear norm minimization if
\begin{equation}\label{m-lowrank-Gaussian}
m \geq C r(d_1+d_2).
\end{equation}
For small $r$, the right hand side is in fact smaller than the dimension $d_1d_2$ of the space of matrices, so that we can solve an underdetermined system of equations.

For more information on low rank matrix recovery, we refer to \cite{davenport2016overview,candes2009exact,gross2011recovering,KuengRauhutTerstiege2017,recht2010guaranteed}.

\begin{remark}
    Whereas low rank matrix recovery generalizes compressive sensing from sparse vectors to low rank matrices in a rather straightforward way, aiming for optimal sample efficiency under mixed priors, e.g., recovering low rank matrices that are simultaneously sparse in their entries, is far more challenging. 
    
    For instance, it has been shown that, in general, linear combinations of convex regularizers for different sparsity structures do not allow to
    outperform recovery guarantees of the “best” one of them alone \cite{oymak2015simultaneously}. In other words, if $\mathbf X \in \R^{d_1\times d_2}$ is of rank $r$ and has only $s_1$ non-zero columns and $s_2$ non-zero rows, a linear combination of $\ell_1$- and nuclear norm can only achieve recovery for $m \gtrsim \min\{ r (d_1+d_2), s_1s_2 \log(e(d_1d_2)/(s_1s_2)) \}$ instead of $m \gtrsim r (s_1+s_2)$ which would be the information theoretic minimum in terms of encoded information; for $r \ll s_1,s_2 \ll d_1,d_2$, this can make a notable difference in the number of required measurements.

    There are only few approaches that come with rigorous performance guarantees. 
    For strongly restricted measurement setups, where the measurement map is constructed in a nested way \cite{bahmani2016near,foucart2020jointly}, (near)-optimal bounds on the number of required measurement can be shown.
     For general measurements, guarantees for 
     certain iterative reconstruction methods are provided in \cite{lee2017near,fornasier2021robust,foucart2020jointly,maly2023robust}, which however, require a suboptimal number of measurements. It is currently unclear whether there is a statistical to computational gap in the sense that tractable algorithms cannot achieve successful recovery with a number of measurements of the same order as the one provably achievable with intractable methods \cite{foucart2020jointly}. Note that recovering matrices that are simultaneously sparse and low rank is inherently related to \emph{Sparse Principal Component Analysis (SPCA)} \cite{zou2006sparse,d2004direct} in statistics, where a similar statistical to computational gap occurs \cite{berthet2013,brennan18a}.

     A further extension of low rank matrix recovery considers the recovery of low rank tensors $\bX \in \R^{n_1\times n_2 \times \cdots \times n_d}$ from linear measurements $y = \mathcal{A}(\bX) \in \R^m$, where $m < n_1n_2 \cdots n_d$. While preliminary results are available, guarantees for the number of required random measurements $m$ to successfully recover $\bX$ are considerably worse for tractable algorithms than for computationally infeasible approaches, see for instance \cite{barak2016noisy,ghadermarzy2019near,mu2014square,rauhut2017low,rauhut2021tensor} and the references therein. Also here, it is currently open whether there is a provable statistical to computational gap.
\end{remark}

\section{Generalization for Learned Unrolled Recovery Algorithms}
\label{sec:Unrolled}

While compressive sensing is part of signal processing, the mathematical insights gathered therein prove valuable for exploring the theoretical foundations of modern machine learning. Vice versa, sparse recovery algorithms
can be enhanced by using machine learning techniques.
To illustrate this, let us first introduce the generic formulation of a \emph{supervised} learning problem, see also \cite{mohri2018foundations,shalev2014understanding}. \\

\textbf{Supervised learning.} Supervised learning assumes access to labeled \emph{training data} $(\bx_i,\by_i)_{i=1}^n$, for $n \in \N$, and seeks to learn a functional relation $f \colon \calX \to \calY$ between \emph{input data} $\bx_i \in \calX$ and \emph{output data} $\by_i \in \calY$ that extends to new, previously not observed data.\footnote{One commonly refers to the entries of $\bx_i$ as \emph{features} and to the vectors $\by_i$ as \emph{labels}.} This is formalized by assuming that $(\bx_i,\by_i)_{i=1}^n$ are independent and identically distributed (iid) samples from an unknown data distribution $\mu$. One says that $f$ generalizes well if $f(\bx) \approx \by$, for $(\bx,\by) \sim \mu$. 

To quantify generalization of $f$, one introduces a loss function $\ell \colon \calY \times \calY \to \R$ measuring the discrepancy between different elements of $\calY$, and defines the corresponding \emph{risk} of $f$ under $\ell$ as
\begin{align*}
    \calR(f) = \E_{(\bx,\by)\sim \mu} \; \ell(f(\bx),\by).
\end{align*}
Good generalization of $f$ compared to an alternative explanation $\hat f$ thus corresponds to a small value of $\calR(f)$ compared to $\calR(\hat f)$.

Since $\calR$ depends on $\mu$, one has in general no access to the risk. A popular substitute is thus the \emph{empirical risk}
\begin{align}
\label{emp:risk}
    \hat \calR_n (f) = \frac{1}{n} \sum_{i=1}^n \ell(f(\bx_i),\by_i),
\end{align}
which can be evaluated on the accessible training data and converges to the risk for $n \to \infty$.

The empirical risk can also be used to identify good explanations $f$ among all possible solutions. To this end, one fixes a \emph{hypothesis class} $\calH \subset \{ f \colon \calX \to \calY \}$ of suitable functions, e.g., polynomials of a certain degree, and then determines the best explanation $h_\star \in \calH$ by solving
\begin{align}
\label{eq:ERM}
    h_\star \in \argmin_{h \in \calH} \hat R_n(h).
\end{align}
To allow efficient numerical implementation of \eqref{eq:ERM}, it is common to parametrize the functions in $\calH$ by real-valued vectors, i.e., to write $\calH = \{ h_{\btheta} \colon \btheta \in \Theta \}$, for a suitable choice $\Theta \subset \R^p$. In the case of polynomials of degree at most $d$, $\btheta$ could encode the $d+1$ necessary coefficients so that $\Theta = \R^{d+1}$. In this case, we would have $h_\star = h_{\btheta_\star}$ with
\begin{align*}
     \btheta_\star \in \argmin_{\btheta \in \Theta} \hat R_n(h_{\btheta}).
\end{align*}
Identifying an optimal $h_\star$ from the $(\bx_i,\by_i)_{i=1}^n$ is called \emph{training}.

To characterize the generalization of a computed solution $h_\star$, it is now crucial to understand the difference $|\hat \calR_n(h_\star) - \calR (h_\star) |$. Since one does not know a priori which function in $\calH$ is picked during training, a natural approach is to control  
\begin{align}
\label{eq:GE}
    \sup_{h\in \calH} |\hat \calR_n(h) - \calR (h) |
\end{align}
for the hypothesis class of interest.\\

\textbf{Neural networks.} In the past decade different variants of deep neural networks have given rise to state-of-the-art hypothesis classes due to their impressive empirical performance. For conciseness, we only discuss the most basic type of neural networks here, namely \emph{fully connected feed-forward networks (FNNs)}. Other popular architectures include residual networks, convolutional, recurrent networks, and transformers. 

An $L$-layer FNN $h \colon \R^{d_0} \to \R^{\dout}$ is defined as
\begin{align}
\label{eq:FNN}
    h(\bx) = \bW^{(L)} \cdot h_{L-1} \circ \cdots \circ h_1(\bx),
    \quad \text{where} \quad
    h_\ell(\bz) = \sigma( \bW^{(\ell)}\bz + \bb^{(\ell)}),
\end{align}
for \emph{weight matrices} $\bW^{(\ell)} \in \R^{d_\ell \times d_{\ell-1}}$ with $d_L = \dout$, \emph{biases} $\bb^{(\ell)} \in \R^{d_\ell}$, and an \emph{activation function} $\sigma \colon \R \to \R$ which is applied entry-wise. Weights and biases are trainable parameters of $h$, the activation function and the dimensions $d_1,\dots,d_{L-1}$ are commonly fixed when defining the hypothesis space $\calH$. Common choices of $\sigma$ include sigmoidal functions and the ReLU $\sigma(z) = \max\{0,z\}$. 
We refer to \cite{petersen2026} for an introduction and more information.



\subsection{
Interpreting recovery algorithms as neural networks}



As many other variational approaches for solving inverse problems, compressive sensing hinges on the idea that prior knowledge on the signal to be recovered (sparsity) allows to make recovery from few measurements possible. In the age of machine learning it is natural to try to refine such prior knowledge by taking into account training samples of possible signals to be recovered. 
One natural way to use deep learning concepts in this context is to interpret iterative recovery algorithms as neural networks \cite{gregor2010learning}. Making certain parameters in these networks learnable allows to put them
into standard machine learning frameworks and to improve performance via adapting to training data. Let us note in passing that a notable alternative is to learn a regularizer (replacing the $\ell_1$-norm in \eqref{eq:BPDN}) that is represented via a deep neural network \cite{maass2019}. We will only focus on the former approach and study corresponding generalization error bounds \cite{schnoor2023generalization} in the sense of the previous chapter.

The authors of \cite{gregor2010learning} 
propose \emph{Learning ISTA (LISTA)} and interpret the first $L$ iterations \eqref{eq:ISTA} of ISTA  
\begin{align*}
    \bx^\ell = \St_{\frac{\lambda^\ell}{2}} \Big( (\bI - \bA^T\bA) \bx^{\ell-1} + \bA^T\by \Big)
\end{align*}
as layers $h_\ell$ of a neural network with activation function $\St_{\frac{\lambda}{2}}$ and initialization $\bx^0 = 0$
\begin{align*}
    h_\ell(\bx,\by) = \St_{\frac{\lambda^\ell}{2}} \Big( \bW_x \bx + \bW_y \by \Big),
\end{align*}
replacing $(\bI - \bA^T\bA)$ and $\bA^T$ by parameter matrices $\bW_x \in \R^{d\times d}$ and $\bW_y \in \R^{d\times m}$, respectively. The resulting $L$-layer FNN $h = h_{L} \circ \cdots \circ h_1 \colon \R^m \to \R^d$,
\[
h(\by) = \corr{h_{L}(h_{L-1}( \cdots h_1(0,\by) \cdots,\by),\by)},
\]
with trainable parameters $(\bW_x,\bW_y,\lambda^1,\hdots,\lambda^L)$ exhibits weight sharing and skip connections from the input $\by$ to each layer. Since $\St_\lambda$ can be decomposed into the sum of two ReLU functions, LISTA is naturally related to ReLU FNNs.

If one is interested in preserving the observation matrix $\bA$ and only seeks to learn an orthonormal basis $\bPhi \in O(d)$ in which the 
signals to be recovered
are sparse, one can alternatively consider
\begin{align}
\label{eq:LISTA_orthogonal}
    h_\ell(\bx,\by) = \St_{\frac{\lambda^\ell}{2}} \Big( (\bI - \bPhi^T\bA^T\bA\bPhi) \bx + \bPhi^T\bA^T\by \Big)
\end{align}
with trainable parameters $(\bPhi,\lambda^1,\hdots,\lambda^L)$ and $\bPhi \in O(d)$ ranging over the orthogonal group \cite{behboodi2022compressive}. 

Following LISTA, various works explored the idea of unfolding iterative thresholding algorithms \cite{wu2020sparse,sun2016deep,liu2019alista,maass2019,behboodi2022compressive,schnoor2023generalization}. Also other iterative optimization algorithms have been re-interpreted as unfolded neural networks, see e.g.\ \cite{ha20,kouni25}.


\subsection{Generalization error bounds}

There are different complexity measures that can be used to quantify the error in \eqref{eq:GE} depending on $\calH$. Notable examples include \corr{Vapnik--Chervonenkis} (VC) dimension \cite{vapnik2015uniform}, Rademacher complexity \cite{koltchinskii2002rademacher,bartlett2002rademacher}, local Rademacher complexity \cite{babome05}, stability \cite{shalev2010learnability,rakhlin2005stability}\corr{,} and robustness \cite{xu2012robustness}. We focus here on the (empirical) Rademacher complexity.

\begin{definition}
    Given a data set $S = (\bx_i,\by_i)_{i=1}^n$ and a hypothesis space $\calH$ of real-valued functions, the \emph{empirical Rademacher complexity} of $\calH$ on $S$ is defined as
    \begin{align*}
        R_S(\calH) = \E_\varepsilon \sup_{h \in \calH} \frac{1}{n} \sum_{i=1}^n \varepsilon_i h(\bx_i),
    \end{align*}
    where $\varepsilon_1,\dots,\varepsilon_n$ are iid Rademacher variables, i.e., $$\Pr(\varepsilon_i = 1)~=~\Pr(\varepsilon_i = -1)~=~\frac{1}{2}.$$
\end{definition}

The following result provides a general bound of \eqref{eq:GE} in terms of the empirical Rademacher complexity of $\ell \circ \calH$, i.e., the class of functions which are a composition of the loss $\ell$ with functions in $\calH$.

\begin{theorem}[{\cite[Theorem 26.5]{shalev2014understanding}}]
\label{thm:Generalization_General}
    Let $\calH$ be a set of functions, let $S = (\bx_i,\by_i)_{i=1}^n$ be a training set drawn from a probability measure $\mu$, and let $\ell$ be a real-valued loss function satisfying $|\ell(h(\bx_i),\by_i)| \le c$, for some $c > 0$ and all $h \in \calH, (\bx_i,\by_i) \in S$. Then, for any $\delta \in (0,1)$, we have with probability at least $1-\delta$ that
    \begin{align*}
        \calR(h) \le \hat{\calR}_n(h) + 2 R_S(\ell \circ \calH) + 4c \sqrt{\frac{2 \log(\frac{4}{\delta})}{n}},
    \end{align*}
    for all $h \in \calH$.
\end{theorem}

Let us consider the specific class of $L$-layer LISTA networks with orthogonal shared weights and functions $h_\ell$ as in \eqref{eq:LISTA_orthogonal} 
\begin{align}
\label{eq:H_LISTA}
    \calH_{\text{LISTA}} = \{ h_\Phi = \sigma \circ \bPhi \circ h_L \circ \cdots \circ h_1 : \Phi \in O(d)\},
\end{align}
i.e.\ the hypothesis \corr{class} of unrolled neural networks parameterized by the sparsifying transform $\Phi \in O(d)$. Here $\bA \in \R^{m\times d}$ is fixed in advance and $\sigma \colon \R \to \R$ with
\begin{align*}
    \sigma(\bz) = \min\{ B, \| \bz \|_2 \} \cdot \frac{\bz}{\| \bz \|_2}
\end{align*}
is a $1$-Lipschitz and bounded non-linearity which is introduced for technical reasons. We consider a fixed parameter $\lambda$ here, but note that we can extend the result below also to learnable thresholding parameters $\lambda_1,\hdots,\lambda_L$ \cite{schnoor2023generalization}.

As a loss function to measure risk and empirical risk we use
the (non-squared) $\ell_2$-loss $\ell(\bz,\bz') = \frac{1}{2} \| \bz - \bz' \|_2$. Given data \corr{$(\bx_i,\by_i=\bA\bx_i)$}, $i=1,\hdots, n$,
one can learn $\Phi \in \corr{O(d)}$ via empirical risk minimization, and in practice via variants of (stochastic) gradient descent using standard toolboxes. Empirically learned unrolled neural \corr{networks} often outperform standard (unlearned) compressive sensing approaches in many applications, see e.g. \cite{behboodi2022compressive,schnoor2023generalization,kouni23}.

The following generalization bound can be derived \cite{behboodi2022compressive}, which provides information on how many training samples are required to achieve \corr{a} small generalization error.

\begin{theorem}
\label{thm:GeneralizationLISTA_simple}
    Let $\bA \in \R^{m\times d}$. Assume that $(\bx_i)_{i=1}^n$ are i.i.d.\ samples of an unknown distribution $\mu$ such that $\|\bx\|_2 \leq B$ almost surely for $\bx \sim \mu$ and set
    $y_i = A x_i$. 
    Consider $\calH_{\text{LISTA}}$ in \eqref{eq:H_LISTA} with $\lambda \in (0,\|\bA \|_{2\to 2}^{-2})$. Then, we have for any $\delta \in (0,1)$ that with probability at least $1-\delta$
    \begin{align*}
        \calR(h) \lesssim \hat{\calR}_n(h) + B \sqrt{\frac{md \log(L^2)}{n}} + B \frac{d \sqrt{\log(L)}}{n} + \sqrt{\frac{\log(\frac{1}{\delta})}{n}},
    \end{align*}
    for all $h \in \calH_{\text{LISTA}}$.
\end{theorem}

A nice feature of the generalization bound in Theorem \ref{thm:GeneralizationLISTA_simple} is that the depth $L$ only enters logarithmically into the error. To deduce Theorem \ref{thm:GeneralizationLISTA_simple} from Theorem \ref{thm:Generalization_General}, one controls the empirical Rademacher complexity of $\ell \circ \calH_{\text{LISTA}}$ via 
Dudley's inequality, see e.g.\ \cite{foucart2013mathematical}. 
In order to bound the appearing covering numbers, one essentially needs to bound the Lipschitz constants of the functions $h_\Phi$
with respect to the parameters $\Phi$ and for fixed input.

Theorem \ref{thm:GeneralizationLISTA_simple} can be extended to more general deep network architectures \cite{schnoor2023generalization} and to other types of unrolled neural networks \cite{kouni25}.
For further discussion on generalization bounds, we refer the reader to \cite{jiang2019fantastic,behboodi2022compressive,schnoor2023generalization,shalev2014understanding,mohri2018foundations,kouni25}.

\section{Implicit Regularization}
\label{sec:implicit-reg}

In many scenarios of deep learning, the number of network parameters exceeds the number of available training samples by far. In such overparameterized regimes, there are usually infinitely many networks interpolating the training data exactly so that the empirical risk has infinitely many global minimizers (and generalization bounds like the ones of the previous section become vacuous).
Many of these interpolating networks will make \corr{incorrect} predictions on new data.
Nevertheless, learned neural networks often still generalize well to unseen data \cite{zhang2017understanding}. This finding challenges the \corr{conventional} wisdom that complicated models \corr{cause} overfitting \cite{hastie2001elements} and \corr{that} model design should reflect a tradeoff between bias and variance \cite{geman1992neural,kohavi1996bias}. In fact, the double descent curve discovered by \cite{belkin2020double} demonstrates that overparameterization does not necessarily lead to \corr{a} large generalization error, as long as the level of overparameterization is sufficiently high, as depicted in Figure~\ref{fig:double-descent}.



\begin{figure}
\centering
\begin{subfigure}[c]{0.48\textwidth}
\includegraphics[width=\textwidth]{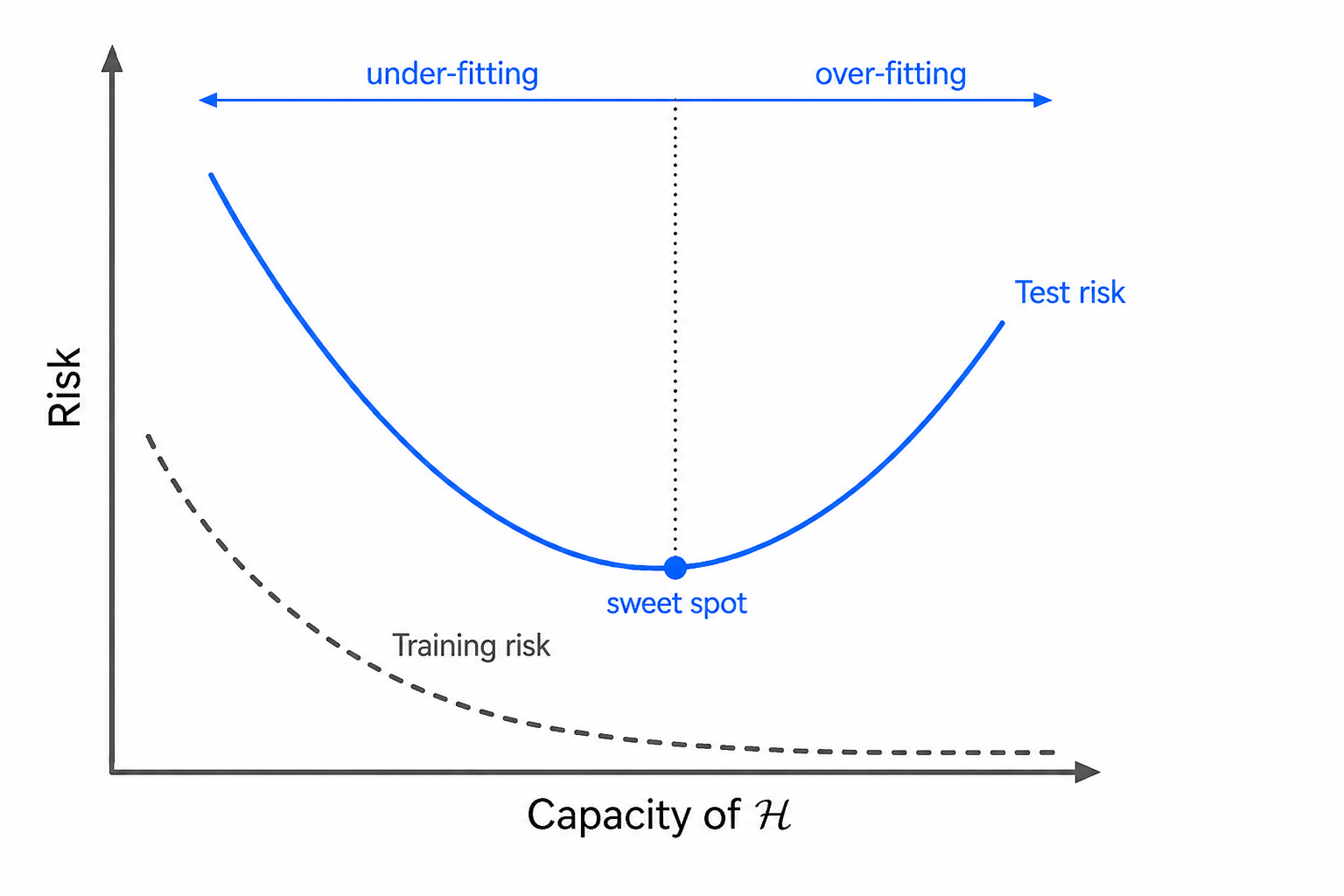}
\subcaption{Classical picture: neither too few nor too many parameters are desirable.}
\end{subfigure}
\hfill
\begin{subfigure}[c]{0.48\textwidth}
\includegraphics[width=\textwidth]{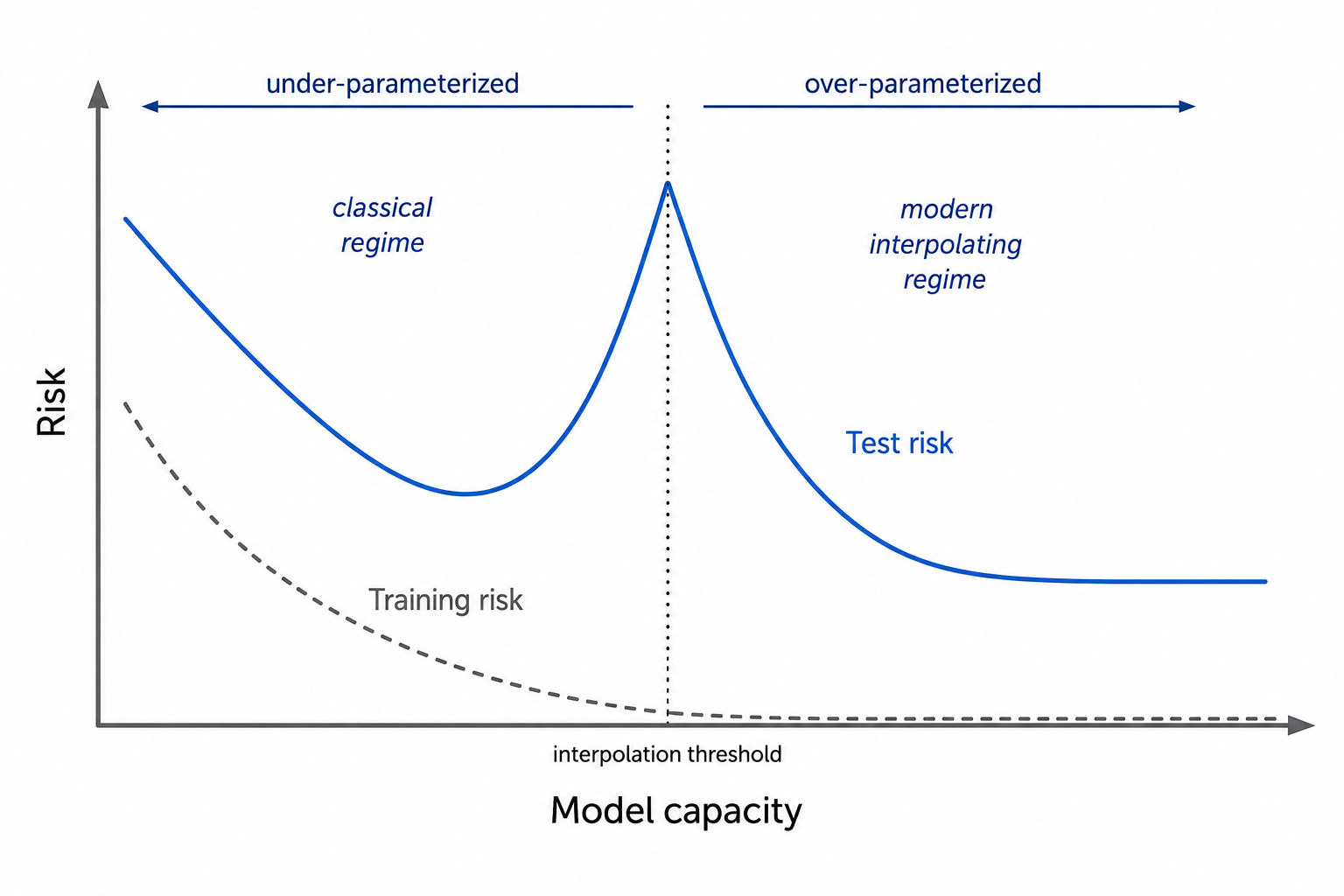}
\subcaption{Double descent curve: Overparameterization beyond the number of training samples can lead to good generalization.}
\end{subfigure}
\caption{Curves for training risk and test risk.}
\label{fig:double-descent}
\end{figure}

In the overparameterized setting, the employed optimization algorithm and its hyperparameters\corr{,} such as initialization\corr{,} have
\corr{a} significant influence on the computed solution.
This is termed implicit regularization (or implicit bias). Apparently, commonly used algorithms such as (stochastic) gradient descent 
have implicit bias towards solutions that generalize well. This phenomenon is at the core of understanding why deep learning works so well, but so far only partial results for simplified models are available on which we report in the following. 
They hint towards the hypothesis that the implicit regularization of (stochastic) gradient descent (with small initialization) leads to models of low complexity such as sparsity of the weight matrices.
This also reveals a perhaps unexpected connection to compressive sensing. \corr{Considerably more research is needed} to understand this phenomenon \corr{for} more realistic deep learning models.

\subsection{Overparameterization and implicit regularization}


For the purpose of this exposition, we will focus on learning neural networks $h=h_{\theta}$ of the form \eqref{eq:FNN}, where $\btheta = (\bW_1,\hdots,\bW_N,\bb_1,\hdots,\bb_N)$ collects the learnable parameters, via gradient descent and its continuous analog gradient flow. Given training data $(\bx_i,\by_i)$, $i=1,\hdots,n$, and a suitable loss function $\ell$, we aim at minimizing the associated empirical risk $\mathcal L(\btheta) := \calR_n (h_\btheta)$ defined in \eqref{emp:risk}. For simplicity, we will restrict ourselves to losses with $\min_\btheta \mathcal L(\btheta) = 0$ here.

Gradient flow starts from a suitable initialization $\btheta(0)$ and considers the solution $\btheta(t)$, $t \in [0,\infty)$, of the ODE
\begin{equation}
\label{def:grad-flow}
\frac{d}{dt} \btheta(t) = - \nabla\mathcal{L}(\btheta(t)).
\end{equation}
Gradient descent can be considered as Euler discretization of this ODE. Starting from $\btheta(0)$ it computes the iterations
\begin{equation}
\label{def:grad-descent}
\btheta(\tdisc+1) = \btheta(\tdisc) - \eta_\tdisc\nabla\mathcal{L}(\btheta(\tdisc)),
\end{equation}
where $\eta_\tdisc>0$ are the step sizes (learning rates).
Gradient flow is often easier to analyze, but in practice only (variants of) gradient descent can be computed. 
In both cases, we are interested in the limit $\btheta_\infty = \lim_{t \to \infty} \btheta(t)$ and $\btheta_\infty = \lim_{\tdisc \to \infty} \btheta(\tdisc)$. Ideally, $\bx_\infty$ reaches a global minimizer of $\mathcal{L}$. In the overparameterized scenario however, a global minimizer will usually not be unique and we would like to characterize which minimizer is selected as $\btheta_\infty$.


The hope for a theoretical analysis of implicit regularization is that a variational principle for the limit $\btheta_\infty$ can be identified of the form
\[
\btheta_\infty \in \argmin_{\btheta} R(\btheta) \quad \mbox{ subject to } L(\btheta) = 0.
\]
The regularizer $R$ may depend on hyperparameter choices such as the initialization $\btheta(0)$ and the step sizes. It describes the implicit regularization and properties of $\btheta_\infty$. For instance, we will later see that $R$ will be close to an $\ell_1$-norm in a specific situation, which implies that gradient descent/flow may promote sparsity as implicit bias.

However, we \corr{do not} expect that implicit regularization can always be characterized by a simple variational problem as above. For a nonlinear neural network with many layers\corr{,} the situation may be much more complicated. Recent research even suggest that neural network generalization is only explainable by the balanced interaction of several different types of implicit regularization of gradient descent \cite{matveev2026conflicting}.

It is still believed that implicit regularization is a key to analyze the success of modern machine learning models, where $\mathcal{L}=0$ does not uniquely define a solution, and hence the algorithm and initialization must be taken into account.


\medskip

Let us illustrate the idea of implicit regularization with a simple example. Consider the least squares loss 
\[
\Lls(\btheta)=\frac{1}{2}\|\bA\btheta-\by\|_2^2
\]
with gradient flow $\btheta(t)$ and initialization $\btheta(0)=\btheta_0$, where $\bA \in \R^{m\times d}$ with $m<d$. 
Then there are infinitely many minimizers of $\Lls$. (If $\bA$ has full rank these are all solutions of 
$\bA \btheta = \by$.)
The following \corr{result} shows that gradient flow 
converges to a limit that can be easily characterized.
In order to state \corr{it,} we denote the Moore--Penrose \corr{pseudoinverse} of $\bA$ by $\bA^\dagger$.
\begin{proposition}
Consider the gradient flow
\begin{equation}\label{eq:lsq_pseudo}
    \dot{\btheta}(t)=-\nabla\Lls(\btheta), \qquad \btheta(0)=\btheta_0
\end{equation}
for $\Lls$.
Then $\btheta(t)$ converges as $t\to\infty$, and its limit is given by
\begin{equation*}
    \btheta_\infty=\bA^\dagger\by+(\bI-\bA^\dagger\bA)\btheta_0.
\end{equation*}
\end{proposition}

\begin{proof}
Define $\bz(t)=\btheta(t)-\bA^\dagger\by$. Since $\bA^\top(\bA\bA^\dagger\by-\by)=0$, we have
\begin{align*}
    \frac{d}{dt}{\bz}(t)
    &=-\bA^\top(\bA\btheta(t)-\by)\\
    &=-\bA^\top(\bA\btheta(t)-\bA\bA^\dagger\by)
    =-\bA^\top\bA\bz(t).
\end{align*}
Therefore
\begin{equation*}
    \bz(t)=\exp(-t\bA^\top\bA)(\btheta_0-\bA^\dagger\by).
\end{equation*}
Since $\bA^\top\bA$ is symmetric and positive semidefinite,
\begin{equation*}
    \lim_{t\to\infty}\exp(-t\bA^\top\bA)=\bI-\bA^\dagger\bA
\end{equation*}
because $\bI-\bA^\dagger\bA$ is the orthogonal projector onto the kernel of $\bA$. Consequently,
\begin{equation*}
    \btheta_\infty=\bA^\dagger\by+(\bI-\bA^\dagger\bA)(\btheta_0-\bA^\dagger\by).
\end{equation*}
Since $\bA^\dagger\by\in\operatorname{range}(\bA^\top)$ and $\bI-\bA^\dagger\bA$ is the orthogonal projection onto $\ker(\bA)$,
\begin{equation*}
    (\bI-\bA^\dagger\bA)\bA^\dagger\by=0.
\end{equation*}
Hence
\begin{equation*}
    \btheta_\infty=\bA^\dagger\by+(\bI-\bA^\dagger\bA)\btheta_0. \tag*{\qed}
\end{equation*}
\end{proof}
In fact, \eqref{eq:lsq_pseudo} can be expressed as
\begin{equation}\label{eq:implicit_bias_onelayer}
    \btheta_\infty \in \argmin_{\bz\in\mathbb{R}^d}\|\bz-\bx_0\|_2^2,\quad
    \text{subject to}\quad \mathcal{L}(\bz)= 0.
\end{equation}

\medskip

Also a different type of implicit regularization may occur where the training dynamics $\btheta(t)$ remains close to certain solutions for longer time intervals before moving quickly to other solutions. In such a situation early stopping may promote implicit regularization and promote certain desirable properties. If one considers only a first phase of the training this phenomenon is sometimes referred to as early alignment \cite{boursier2024}. We will discuss a similar phenomenon in the very simplified setting of matrix regression below \cite{chou2024gradient}.

We note that below we will focus on the case of small initialization, where one observes so-called feature learning. For larger (random) initialization and for very wide networks the training dynamics is usually simpler in the sense that the weights do not move far from their initialization. In this situation, the dynamics can be approximated by one in a linear reproducing kernel Hilbert space generated by the so-called neural tangent kernel \cite{jacot2018ntk}. The latter is sometimes referred to as the lazy regime \cite{ch19}, while the former (to be discussed below in more detail) refers to the rich regime \cite{wo20}.

\medskip

It seems daunting to analyze implicit regularization for gradient descent applied to realistic deep neural networks directly. For a mathematical analysis of the basic principles underlying implicit regularization it is useful to first reduce to simplified minimization problems which still show the same characteristics, namely, (i) infinitely many global minimizers, and (ii) a factorization or compositional structure as in deep neural networks.

As a first example, we discuss an underdetermined linear inverse problem (as in compressive sensing) where the unknown is factorized as a linear neural network.
As a slightly more advanced model, we pass to fully connected linear networks, which, however, are already significantly more difficult to analyze and where basic problems are still open.

\subsection{Diagonal linear networks}

We consider the problem of minimizing the rather simple function
\begin{equation}\label{def:ls}
\mathcal{K}(\bx) = \frac{1}{2} \|\bA \bx - \by\|_2^2, \quad \bx \in \R^d,
\end{equation}
where $\bA \in \R^{m \times d}$ with $m < d$ and 
$\by \in \corr{\R^m}$.  Obviously, $\mathcal{K}$ has infinitely many global \corr{minimizers}, namely all solutions to the underdetermined system $\bA \bx = \by$.\footnote{Even when $\bA$ does not have full rank and this system does not have a solution, then \corr{$\mathcal{K}$} still has infinitely many global minimizers given by $\bx = \bA^\dagger \by + \bv$ with $\bv \in \operatorname{ker}(\bA)$ and $\bA^\dagger$ being the \corr{pseudoinverse} of $\bA$.} If \corr{$\by = \bA \bx_0$} for a sparse $\bx_0$\corr{,} the problem of finding a sparse global minimizer of $\mathcal{K}$ is a compressive sensing problem as discussed in \corr{Section}~\ref{sec:Basics-CS}.

In order to inject a neural network structure into this problem we factorize $\bx$ as 
\begin{equation}
\label{x-factorization}
\bx = \bx^{(L)}\odot\cdots\odot\bx^{(1)},
\end{equation}
where $(\bv \odot \bw)_j = \bv_j \bw_j$, $j=1,\hdots,d$, denotes the Hadamard product of two vectors $\bv,\bw \in \R^d$.
Note that the Hadamard product may be written in terms of matrix multiplication of diagonal matrices
\[
\diag(\bx) = \diag(\bx^{(L)}) \cdots \diag(\bx^{(1)}),
\]
which is the reason why the above Hadamard factorization is also referred to as linear diagonal neural network, i.e.,
a neural network of the form \eqref{eq:FNN} with identity as activation function $\sigma$, the weight matrices of the form $\bW^{(\ell)} = \diag(\bx^{(l)})$ and zero biases $\bb^{(l)}=0$. 

Combining all vectors \corr{$\bx^{(1)},\hdots,\bx^{(L)}$} into a parameter vector $\btheta$ and plugging the factorization \eqref{x-factorization} into the function \corr{$\mathcal{K}$,} we obtain
\begin{align*}
    \mathcal{L}(\btheta) &= \mathcal{L}(\bx^{(1)},\hdots,\bx^{(L)})
    = \mathcal{K}(\bx^{(L)} \odot \cdots \odot \bx^{(1)}) \\
    &= \frac{1}{2} \|\bA (\bx^{(L)}\odot \cdots \odot \bx^{(1)}) - \by\|_2^2.
\end{align*}
Note that, in the following, we will distinguish between a target loss $\mathcal K$ and its (overparameterized) realization $\mathcal L$ that is minimized via gradient descent/flow.
We consider minimization $\mathcal{L}$ via gradient flow \eqref{def:grad-flow}, i.e.,
\begin{equation}
\label{grad-flow-diag}
\frac{d}{dt} \bx^{(\ell)}(t) = - \nabla_{\bx^{(\ell)}} \mathcal{L}(\bx^{(1)}(t),\hdots,\bx^{(L)}(t)).
\end{equation}
For our analysis we assume that all components $\bx^{(\ell)}$ are initialized identically, i.e.,
\begin{equation}
\label{diag-init}
\bx^{(\ell)}(0) = \alpha \1, \quad \ell = 1,\hdots, L,
\end{equation}
where $\alpha > 0$ is a (small) scalar.
We are interested in the behaviour of the product
\begin{equation}
\label{prod-diag}
\xprod(t) = \bx^{(L)}(t)\odot\cdots\odot\bx^{(1)}(t)
\end{equation}
as $t$ tends to $\infty$. Perhaps surprisingly, 
the next result \cite[Theorem 2.1]{chou2023more} states that $\bx_\infty = \lim_{t \to \infty} \xprod(t)$ approximates an $\ell_1$-minimizing solution of $\bA \bx = \by$ provided the initialization is small enough. In other words, the implicit regularization
of gradient flow on $\mathcal{L}$ is towards sparse solutions.

\begin{theorem}\label{thm:dln_ell_1}
    Let $L\geq 2$, $\bA\in\mathbb{R}^{m\times d}$ and $\by\in\mathbb{R}^{m}$. Assume that $S_+ = \{\bz\geq \0:\bA\bz = \by\}$ is non-empty. Consider the gradient flow \eqref{grad-flow-diag} with initialization \eqref{diag-init} for some $\alpha > 0$, and the product flow $\xprod(t)$ in \eqref{prod-diag}.
Then the limit $\xprodinfty:= \lim_{t\to\infty}
    \xprod(t)$
    exists and $\xprodinfty \in S_+$. Moreover, if $Q:=\min_{\bz\in S_+}\|\bz\|_1>\|\xprod(0)\|_1$, then 
    \begin{equation}\label{eq:L1min_Q}
        \|\xprodinfty\|_1-Q\leq\epsilon Q
    \end{equation}
    where $\epsilon$ is defined as
    \begin{align}\label{eq:L1min_general}
        \epsilon:=\begin{cases}
            \frac{\log(d)}{\log\left(\frac{Q}{d\alpha^2}\right)} &\text{if }L=2,\\[6pt]
            \frac{L(d^\gamma-1)}{ 2((Q/\alpha^L)^\gamma-d^\gamma)}&\text{if }L>2,
        \end{cases}
    \end{align}
    where $\gamma=1-\frac{2}{L}$. Note that $\epsilon$ goes to zero as $\alpha$ goes to zero.
\end{theorem}

A more general result for arbitrary positive initialization was shown in \cite{chou2023more}. \corr{Since} this requires weighted $\ell_1$-\corr{minimization}, we decided to state the simpler version above. The restriction of all relevant quantities in Theorem \ref{thm:dln_ell_1} to the positive orthant is caused by the product structure in \eqref{x-factorization} together with the positive identical initialization of all factors in \eqref{diag-init}. We will discuss below how this can be lifted by slightly changing the overparametrization or the initialization.

\begin{figure}[b]
    \centering
    \includegraphics[width = 0.55\textwidth]{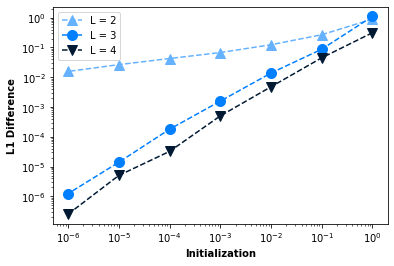}
    \caption{Smaller initialization leads to stronger bias/regularization. Shallow ($L=2$) and deep structure ($L>2$) exhibit different scalings.}
    \label{fig:alpha}
\end{figure}


Note that the remarkable phenomenon of convergence to $\ell_1$-minimizers only occurs when the network has more than one layer, while the \corr{one-layer} case corresponds to $\Lls$\corr{,} which has implicit regularization towards minimal $\ell_2$ rather than $\ell_1$ according to \eqref{eq:implicit_bias_onelayer}. Hence, given that $\bA$ satisfies the null space property, deep diagonal linear \corr{networks} have an implicit regularization towards sparse solutions even though such \corr{a} preference is not obvious from the optimization problem. 

Theorem \ref{thm:dln_ell_1} also explains why the depth of the neural networks plays an important factor, even when the function class represented by the models are identical for all depth (since they are all linear models). Indeed the scaling of $\epsilon$ in $\alpha$ is different for $L=2$ (the shallow case) 
and $L>2$ (corresponding to deep networks),
as demonstrated in 
Figure~\ref{fig:alpha}. 

The proof of Theorem \ref{thm:dln_ell_1} heavily relies on the observation that the trajectory of $\xprod$ approaches the solution set $S_+$ in a particular way, as stated in Lemma \ref{lemma:Bregman_nonincreasing} below. It makes use of the Bregman divergence which, for a strictly convex function $F$ on a subset of $\R^d$, is defined as
\[
D_F(\bz,\bx):=F(\bz)-F(\bx)-\langle\nabla F(\bx),\bz-\bx\rangle.
\]
It is a nonsymmetric generalization of the notion of distance.
Depending on the number of layers, we will use the function $F : \R^d_+ \to \R$, defined as
\begin{equation}\label{eq:Breman_positive}
    F(\bx) 
    = \begin{cases}
    \frac{1}{2} \langle \bx \odot \log(\bx) - \bx, \1 \rangle & \text{if }L = 2\\
    \frac{L}{2(2-L)} \langle \bx^{\odot \frac{2}{L}}, \1 \rangle & \text{if }L > 2.
    \end{cases} 
\end{equation}
The next result states that the rate of decrease of the Bregman divergence between $\xprod$ and any element $\bz \in S_+$ is independent of $\bz$. We will include the proof here because it highlights the crucial ideas that persistently appear also in results that we will present later.

\begin{lemma}[\cite{chou2023more}, Lemma 2.6]\label{lemma:Bregman_nonincreasing}
Let $\bx^{(\ell)}(t)$ follow the gradient flow \eqref{grad-flow-diag} initialized with $\bx^{(\ell)}(0) = \alpha\1$, $\ell=1,\hdots,L$, for $\alpha > 0$. Then, for any $\bz\in S_+$, 
\begin{equation}\label{eq:DFderivative}
    \frac{d}{dt} D_F(\bz,\xprod(t))
    =-2L\cdot\mathcal{L}(\bx^{(1)}(t),\cdots,\bx^{(L)}(t)).
\end{equation}
\end{lemma}
\begin{proof}
    If $\bx^{(1)} = \cdots = \bx^{(L)}$ then it holds
    \[
    \nabla_{\bx^{(1)}} \mathcal{L}(\bx^{(1)},\hdots,\bx^{(L)}) = \cdots = \nabla_{\bx^{(L)}} \mathcal{L}(\bx^{(1)},\hdots,\bx^{(L)}).
    \]
    Due to identical initialization $\bx^{(1)}(0) = \cdots \bx^{(L)}(0) = \alpha\1$ this implies that $\bx(t):=\bx^{(1)}(t) = \cdots = \bx^{(L)}(t)$,
     and consequentially $\xprod(t)=\bx^{(L)}(t)\odot\cdots\odot\bx^{(1)}(t)=\bx^{\odot L}(t)$. Due to continuity and since $\nabla_{\bx ^{(\ell)}} \mathcal{L}(0,\hdots,0)  = 0$, we have $\bx(\tcont) \geq 0$ for all $\tcont\geq 0$. For any $\bz\in S_+$ it holds
\begin{align*}
    & \partial_\tcont D_F(\bz,\xprod(\tcont))
    = \partial_\tcont \left[ F(\bz) - F(\xprod(\tcont)) - \langle \nabla F(\xprod(\tcont)), \bz-\xprod(\tcont) \rangle \right]\\
    &= - \langle \nabla F(\xprod(\tcont)), \xprod'(\tcont) \rangle 
    - \langle \partial_\tcont \nabla F(\xprod(\tcont)), \bz-\xprod(\tcont) \rangle
    + \langle \nabla F(\xprod(\tcont)), \xprod'(\tcont) \rangle\\
    &= - \langle \partial_\tcont \nabla F(\xprod(\tcont)), \bz-\xprod(\tcont) \rangle.
\end{align*}
By the chain rule and since the Hessian $\mathbf H_F = \frac{1}{L}\operatorname{\diag}(\bx^{\otimes(-2+\frac{2}{L}}))$ is diagonal, we obtain
\begin{align*}
    & \partial_t \nabla F(\xprod(\tcont))
    = \mathbf H_F (\xprod(\tcont)) \cdot \xprod'(\tcont)
    = \frac{1}{L} \xprod(\tcont)^{\odot(-2+\frac{2}{L})} \odot \xprod'(\tcont)\\
    &= \frac{1}{L} \bx(\tcont)^{\odot(-2L+2)} \odot \big( L\bx(\tcont)^{\odot(L-1)} \odot \bx'(\tcont) \big) \\
    &= \bx(\tcont)^{\odot(-2L+2)} \odot  \bx(\tcont)^{\odot(L-1)} \odot \big( - L\left[\bA^{T}( \bA\bx^{\odot L}(\tcont) -\by)\right] \odot \bx^{\odot L-1}(\tcont) \big) \\
    &= -L\left[\bA^{T}( \bA\xprod(\tcont) -\by)\right].
\end{align*}
Therefore,
\begin{align*}
    \partial_\tcont D_F(\bz,\xprod(\tcont))
    &= L\langle \bA^{T}( \bA\xprod(\tcont) -\by), \bz-\xprod(\tcont) \rangle
    = -L \langle \bA\xprod(\tcont) -\by, \bA\xprod(\tcont) - \bA\bz \rangle\\
    &= - L\|\bA\xprod(\tcont) -\by\|_2^2
     = -2L\cdot \mathcal{K}(\xprod(\tcont)).
\end{align*}
This completes the proof. \hfill \qed
\end{proof}

\begin{remark}
\label{rem:Reparameterization}
    The proof of Lemma \ref{lemma:Bregman_nonincreasing} reveals that, due to identical initialization, all factors $\bx^{(\ell)}$ evolve jointly as $\bx^{(\ell)}(t) = \xprod(t)^{\odot \frac{1}{L}}$. Consequently, the overparameterized gradient flow on $\mathcal L$ with identically initialized factors mimics a gradient flow on the reparameterized regression functional
    \begin{align*}
        \mathcal L_{rp} (\bx) := {\frac{1}{L}}\mathcal K (\bx^{\odot L}) = \frac{1}{2L} \| \bA\bx^{\odot L} - \by \|_2^2. 
    \end{align*}
    without overparameterization.
\end{remark}

\begin{proof}[Sketch of Theorem \ref{thm:dln_ell_1}]
    According to \eqref{eq:DFderivative}, $D_F(\bz,\xprod(t))$ will decrease at the same rate for all $\bz\in S_+$, and will only stop once $\xprod(t)$ arrives at $S_+$ when the right hand side of \eqref{eq:DFderivative} becomes zero. Since $D_F(\bz,\xprod(t))=0$ if and only if $\bz=\xprod(t)$, there exists $\bz^*\in S_+$ such that $D(\bz^*,\xprod(t))=0$. Because $D\geq0$, $\bz^*\in\argmin_{\bz\in S_+} D_F(\bz,\xprod(t))$. According to \eqref{eq:DFderivative}, $D_F(\bz,\xprod(t))-D_F(\bz,\xprod(0))$ is the same for all $\bz$, and hence one can use non-negativity of $D_F$ to argue that  \begin{equation}\label{eq:optimal_x_positive}
        \xprodinfty
        \in\argmin_{\bz\in S_+} D_F(\bz,\xprod(t))
        =\argmin_{\bz\in S_+} D_F(\bz,\xprod_0).
    \end{equation}
    Equation \eqref{eq:optimal_x_positive} gives an explicit characterization of the limit. The final expression then can be derived with a few basic inequalities. \hfill \qed
\end{proof}

Note that \eqref{eq:optimal_x_positive} is expressed in a general form, stating that the dynamics converges to the solution that is closest to the initialization in a suitable sense. 
While this statement may seem intuitive, the identification of the right notion of (Bregman) distance (i.e., of the function $F$) is non-trivial, and for other types of neural networks and loss functions it may be challenging to understand whether this is possible. 


\medskip

{\textbf{Recovery of negative entries.}}
The proof of Lemma~\ref{lemma:Bregman_nonincreasing} revealed that gradient flow on linear diagonal networks of the form 
\eqref{x-factorization} with positive identical initialization 
cannot recover vectors with negative entries. In order to circumvent this problem one may pass to a difference of two factorizations of the form \eqref{x-factorization}, i.e.,
\begin{equation}\label{diff-diag}
\bx = \bu^{(L)} \odot \cdots \odot \bu^{(1)} - \bv^{(L)} \odot \cdots \odot \bv^{(1)}
\end{equation}
and consider the loss function
\[
\mathcal{L}_{\pm}(\bu^{(1)}, \hdots, \bu^{(L)},\bv^{(1)},\hdots,\bv^{(L)}) = \mathcal{K}(\bu^{(L)} \odot \cdots \odot \bu^{(1)} - \bv^{(L)} \odot \cdots \odot \bv^{(1)}).
\]
For identical initialization
\corr{$\bu^{(1)}(0) = \cdots = \bu^{(L)}(0) = \bv^{(1)}(0) = \cdots = \bv^{(L)}(0) = \alpha \1$}, the gradient flow for $\mathcal{L}_{\pm}$ converges to a limit such that the $\ell_1$-norm of the product \corr{$\xprod_\infty = \lim_{t \to \infty}\bigl(\bu^{(L)}(t) \odot \cdots \odot \bu^{(1)}(t) - \bv^{(L)}(t) \odot \cdots \odot \bv^{(1)}(t)\bigr)$} is close to the $\ell_1$-norm of the $\ell_1$-minimizer with possibly \corr{negative} entries. In other words, Theorem~\ref{thm:dln_ell_1} basically extends to linear diagonal networks of the form \eqref{diff-diag} with $S_+$ replaced by $S = \{\bz \in \R^d : \bA \corr{\bz} = \by\}$\corr{;} see \cite[Theorem~1.1]{chou2023more} for details.

When assuming in addition that the matrix $\bA$ satisfies the stable null space property, i.e., the stable and robust null space property in Definition~\ref{def:srNSP} with constant $\rho \in (0,1)$ and arbitrary constant $\tau > 0$, then one can show also the sparse recovery result stated next \cite{chou2023more}. In other words, gradient flow on the factorization \eqref{diff-diag} can also be used as an algorithm for compressive sensing.
\begin{theorem}
Let $\bA \in \mathbb{R}^{m \times d}$ satisfy the stable null space property of order $s$ with constant $\rho\in(0,1)$.
Let $\bx_* \in \mathbb{R}^{d}$ and $\by = \bA\bx_*$.
Let 
\[
\xprod_\infty = \lim_{t \to \infty} \corr{\bigl(\bu^{(L)}(t) \odot \cdots \odot \bu^{(1)}(t) - \bv^{(L)}(t) \odot \cdots \odot \bv^{(1)}(t)\bigr)}
\]
be the limit vector of gradient flow for $\mathcal{L}_{\pm}$ with initialization $\bu^{(\ell)}(0) = \bv^{(\ell)}(0) = \alpha\1$, where $\alpha > 0$ for $\ell=1\hdots,L$.
   Then the reconstruction error satisfies
    \begin{align*}
        \|\xprod_\infty - \bx_*\|_1
        \leq  \frac{1+\rho}{1-\rho}\left( 2\sigma_{s}(\bx_*)_1 + \epsilon\right).
    \end{align*}
where $\epsilon$ is defined in terms of the initialization scale $\alpha$ as in
\eqref{eq:L1min_general}.
    Clearly, the right hand side equals $\epsilon$ if $\bx_*$ is $s$-sparse.
\end{theorem}

\textbf{Acceleration.} The above analysis for gradient flow on linear diagonal networks provides remarkable insights on implicit regularization. In practice, however, one cannot implement gradient flow, but rather a discrete variant such as gradient descent. 
Extensions of the above results for $L = 2$ layers to gradient descent and stochastic gradient descent, which also discuss conditions on the step sizes have been derived in \cite{even2023s}. The general case of $L\geq 2$ layers is studied with a different approach in \cite{bara26}.

\medskip

While it is remarkable that gradient methods exhibit implicit regularization towards sparse solutions, and hence can be seen as alternatives to existing compressed sensing methods, they may not necessarily be competitive in terms of speed and numerical stability.
One important observation is the following: although the theory states that implicit $\ell_1$-regularization becomes more pronounced the closer the initialization is to zero, smaller initialization leads to slower convergence because it takes longer time to escape the origin, which is actually a saddle point of $\mathcal{L}$.

Interestingly enough, one can mitigate this issue by weight normalization. For simplicity, let us directly consider the reparameterized loss $\mathcal L_{rp}$ from Remark \ref{rem:Reparameterization},
\begin{equation}\label{eq:factorized_loss}
    \mathcal{L}_{rp} (\bx) = \frac{1}{2L}\|\bA\bx^{\odot L} - \by\|_2^2,
\end{equation}
which induces the same gradient flow as the one for linear diagonal linear networks 
with identical initialization studied above. 

\emph{Weight normalization} refers to the additional overparameterization
\[
\bx = \frac{r}{\|\bu\|_2}\bu, \quad r \in \R, \bu \in \R^d.
\]
Plugging this into $\mathcal{L}$ we obtain
a loss function on $r$ and $\bu$ as
\begin{equation}\label{eq:factorized_loss_weight_normalization}
    \Lwn(r,\bu) = \mathcal{L}_{rp} \left(\frac{r}{\|\bu\|_2}\bu\right), \quad r \in \R, \bu \in \R^d.
\end{equation}
We consider a slightly modified version of
gradient flow for finding a minimizer of $\Lwn$. Let $(r(\tcont),\bu(\tcont))$ be a solution of the ODE
\begin{align}
    \frac{d}{dt}r(t) & = - \lrr \,\nabla_{r}\Lwn(r(t),\bu(t)),\quad r(0) = r_0,\label{eq:dgdt} \\
    \frac{d}{dt}\bu(t) & = -\nabla_{\bu}\Lwn(r(t),\bu(t)),\quad\bu(0) = \bu_0, \label{eq:dwdt}
\end{align}
where $\lrr > 0$ is called the learning rate ratio. The choice of $\lrr$ is important as revealed by the next result.

\begin{theorem}[\cite{chou2024robust}, Theorem 2.2]\label{theorem:optimality_constant_rate}
Let $L\in\mathbb{N}$, $L\geq 2$, $\bA\in\mathbb{R}^{m\times d}$, $\by\in\mathbb{R}^{m}$ and $\lrr>0$. Suppose $(r(\tcont),\bu(\tcont))$ follow the dynamics in \eqref{eq:dgdt} and \eqref{eq:dwdt} with $r_0,\bu_0>0$ satisfying $r_0\bu_0=\alpha\1$, $\|\bu_0\|_2 = 1$, and 
\[
\lrr\leq r_0^2\leq\|\bA^\dagger\by\|_{2}^{2/L}.
\]
Denote $\bx(\tcont)=\frac{r(\tcont)}{\|\bu(\tcont)\|_2}\bu(\tcont)$. Suppose $S_+ = \{\bz\geq \0:\bA\bz = \by\}$ is non-empty, and that the limit $\xprodinfty:= \lim_{\tcont\to\infty}\bx(\tcont)^{\odot L}$ exists. Define the so-called magnification factor as
\begin{equation}
\label{eq:rho}
    \rho:=\frac{r_0}{\|\bA^\dagger\by\|_{2}^{1/L}}\exp\left(\frac{\|\bA^\dagger\by\|_{2}^{2/L} - r_0^2}{2\lrr}\right).
\end{equation}
Then
\begin{enumerate}
    \item The loss defined in \eqref{eq:factorized_loss} decreases exponentially in time, i.e., for all $\tcont\geq 0$
    \begin{equation}\label{WN:conv:rate}
        \mathcal{L}_{rp}(\bx(\tcont)) \leq \mathcal{L}_{rp}(\bx_0)e^{-c\tcont}    
    \end{equation}
    for some constant $c>0$. In addition, the limit $\xprodinfty$ lies \corr{in} $S_+$.
    \item It holds that $\rho\geq 1$, and the limit $\xprodinfty$ satisfies $\|\xprodinfty\|_1-Q\leq\epsilon Q$ with $\epsilon$ as defined in \eqref{eq:L1min_general} but $\alpha$ replaced with $\rho^{-1}\alpha$, i.e., for $\gamma = 1 - \frac{2}{L}$
    \[
\epsilon=\begin{cases}
            \frac{\log(d)}{\log\left(\frac{Q}{d(\alpha/\rho)^2}\right)} &\text{if }L=2,\\[10pt]
            \frac{L(d^\gamma-1)}{ 2\left(\left(\frac{Q}{(\alpha/\rho)^L}\right)^\gamma-d^\gamma\right)}&\text{if }L>2.
            \end{cases}
    \]
    \end{enumerate}
\end{theorem}

Theorem \ref{theorem:optimality_constant_rate} basically states that with weight normalization, the ``effective initialization'' that controls the implicit regularization becomes $\rho^{-1}\alpha$ instead of $\alpha$. Hence if $\rho$ is sufficiently large, the effective initialization can be small without requiring the true initialization $\alpha$ to be small.

Note that by choosing even moderately small $\lrr$, say $\lrr=0.1$, $\rho$ is already quite large due to the exponential term in \eqref{eq:rho}.
In fact, choosing very small $\lrr$ will cause numerical instabilities and hence is not advised. Empirically, there \corr{is a range} of $\lrr$ such that the algorithm works efficiently without requiring small (true) initialization.


Although Theorem~\ref{theorem:optimality_constant_rate} establishes exponential decay of the loss, the rate constant $c$ may become arbitrarily small in degenerate regimes, for example when some coordinates of the trajectory approach zero. The following result provides a more general convergence guarantee that does not require the trajectory to remain uniformly bounded away from the boundary of the nonnegative orthant, although it yields only a sublinear $O(1/t)$ rate.

\begin{theorem}[\cite{chou2025nnls}, Theorem~2.1]\label{thm:non-Accelerated}
Let $L\geq 2$, $\bA\in\mathbb{R}^{m\times d}$, $\by\in\mathbb{R}^{m}$, and $\bx_0>\0$. Let $\bx^{(\ell)}$ follow the gradient flow \eqref{grad-flow-diag} with initialization $\bx^{(\ell)}(0)=\bx_0$, $\ell=1,\hdots,L$. Define the product $\xprod(t)$ as in \eqref{prod-diag}. Then, for any $\bx_+\in\argmin_{\bz\geq\0}\mathcal{K}(\bz)$, there exists a constant $C>0$, depending only on the distance between $\xprod(0)$ and $\bx_+$, such that
\begin{equation*}
\mathcal{K}(\xprod(\tcont))-\mathcal{K}(\bx_+)\leq\frac{C}{\tcont} \quad \mbox{ for all } \tcont > 0.
\end{equation*}
\end{theorem}

This convergence rate can be improved using momentum-type acceleration. Since the loss induced by the diagonal reparameterization is nonconvex, however, standard acceleration results for convex optimization do not apply directly. Nevertheless, a suitable modification of the continuous-time momentum dynamics accelerates the convergence of the product $\xprod(\tcont)$ from $O(1/t)$ to $O(1/t^2)$. The proof is based on a higher-order Lyapunov function.

\begin{theorem}[\cite{chou2025nnls}, Theorem~2.4]\label{thm:Accelerated}
Let $L\geq 2$, $\bA\in\mathbb{R}^{m\times d}$, $\by\in\mathbb{R}^{m}$, and $\bx_0>\0$. Let $\pmb{\xi}(\tcont)\in\mathbb{R}^d$ be such that $(\pmb{\xi}(\tcont),\xprod(\tcont))$ satisfies
\begin{equation*}
\begin{aligned}
\frac{d}{d\tcont}\pmb{\xi}(\tcont)&=-\frac{\tcont}{2}\pmb{\xi}(\tcont)^{\odot q}\odot\nabla\mathcal{K}(\xprod(\tcont)),\\
\frac{d}{d\tcont}\xprod(\tcont)&=\frac{2}{\tcont}\bigl(\pmb{\xi}(\tcont)-\xprod(\tcont)\bigr),
\end{aligned}
\end{equation*}
with $\pmb{\xi}(0)=\xprod(0)=\xprod_0>\0$, where $q=2-\frac{2}{L}\in[1,2)$ and $\mathcal{K}$ is the least-squares loss defined in \eqref{def:ls}. Then, for any $\bx_+\in\argmin_{\bz\geq\0}\mathcal{K}(\bz)$, there exists a constant $C>0$, depending only on $q$ and the distance between $\xprod_0$ and $\bx_+$, such that
\begin{equation*}
\mathcal{K}(\xprod(\tcont))-\mathcal{K}(\bx_+)\leq\frac{C}{\tcont^2} \mbox{ for all } \tcont > 0.
\end{equation*}
\end{theorem}

\medskip


{\textbf{Networks with nonlinearity.}}
The results on implicit regularization have been extended \cite{chou2024induce} to learning certain neural networks involving nonlinearities, also called generalized linear functions. For a matrix $\bA \in \R^{m \times d}$ and two differentiable activation functions $\sigmaout,\sigmainn : \R \to \R$ acting componentwise, consider functions of the form
\[
g : \R^d \to \R^m, \quad g(\bx) = \sigmaout(\bA(\sigmainn(\bx))).
\]
Given a loss function $\ell: \R^m \times \R^m \to \R$ that is continuously differentiable and convex in the first argument, and a vector $\by \in \R^m$ we consider
\[
\Lgl(\bx):= \ell\Big( \sigmaout \big( \bA\sigmainn(\bx) \big), \by \Big),
\]
which relates to \emph{generalized linear models (GLMs)} \cite{dobson2018introduction} when $\ell$ is chosen to be the square loss.

Let $\bx'(t)$ follow the gradient flow for $\Lgl$ initialized at $\bx_0 \in \mathbb{R}^d$ i.e.,
    \begin{align}\label{eq:gd_IRERM}
       \bx'(t) = -\nabla\Lgl(\bx(t)), \quad \bx(0)=\bx_0.
    \end{align}
We consider then the reparameterized flow 
\begin{equation}\label{reparam-x}
\widetilde{\bx} (t) = \sigmainn(\bx (t))
\end{equation}
and study the implicit regularization of \corr{its} limit (provided it exists). Note that $\sigmainn(t) = t^L$, $\sigmaout(t) = t$\corr{,} and $\ell$ being the square loss reduce to the setting of Theorem~\ref{thm:dln_ell_1} studied above.
We provide here an informal version of the two main theorems of \cite{chou2024induce} with adopted notation, omitting several technical details.

\begin{theorem}[\cite{chou2024induce}, Theorems 2.4 \& 2.7]\label{theorem:vector_IRERM}
    Assume that $\ell$ is continuously differentiable and convex in the first argument, with $\min\ell = 0$,
    and $\ell(\bz,\bz') = 0$ if and only if $\bz = \bz'$. Consider the gradient flow $\bx(t)$ in \eqref{eq:gd_IRERM} and its reparameterization $\widetilde{\bx}(t) =  \sigmainn(\bx (t))$.
Define the reparametrized flow $\widetilde{\bx} (t) := \sigmainn(\bx (t))$. 
Assume that $\sigmainn: \R \to \R$ is invertible and that there exist a continuous antiderivative $h$ of
    \begin{align}\label{eq:hprimedefinitionWithoutPositivity}
        h'(z) 
        := \big( [\sigmainn^{-1}]'(z) \big)^2,
    \end{align}
    and 
    a continuous antiderivative $H$ of $h$. Under certain technical assumptions, $\widetilde{\bx}$ converges to
    \begin{equation}\label{eq:implicit_regularization_Bregman}
        \widetilde{\bx}_\infty \in \argmin_{\Lgl(\sigmainn^{-1}(\bz)) = 0 %
        } D_{F}(\bz,\widetilde{\bx}_0)
    \end{equation}
    where $D_{F}(\bz,\widetilde{\bx}_0)$ is the Bregman divergence with respect to $F(\bz) := \langle \1, H(\bz)\rangle$.    
\end{theorem}
In other words, 
the gradient flow $\Lgl$ exhibits an implicit regularization towards a solution minimizing $D_{F}(\widetilde\bz,\widetilde{\bx}_0)$, where the Bregman divergence $D_F$ is generated by $F(\bz) := \langle \1, H(\bz)\rangle$.

\medskip

{\textbf{Implicit regularization in classification problems.}} Diagonal linear networks have also been investigated as a model problem for training via gradient flow and gradient descent in the context of classification problems. Here implicit regularization towards sparse solutions can be observed as well. We refer to \cite{soudry2018implicit,gunasekar2018characterizing,gunasekar2018convolutional,moroshko2020implicit,yun2021unifying,lyu2020margin,wo20} for details.


\subsection{Fully connected linear networks}

In this section, we pass to the more advanced, but still simplified model of fully connected linear networks. These are neural networks \eqref{eq:FNN}, where the activation function is the identity and the bias terms $\bb^{(\ell)}$ are zero, and the weight matrices $\bW^{(\ell)} \in \R^{d_\ell \times d_{\ell}-1}$ are fully populated -- in contrast to the previous section where they were diagonal. In other words, a fully connected linear network is of the form
\[
h(\bx) = \bW^{(L)} \bW^{(L-1)} \cdots \bW^{(1)} \bx = \bW \bx,
\]
where 
\begin{equation}\label{lin-network}
\bW  = \bW^{(L)} \cdots \bW^{(1)} \in \R^{d_L \times d_0}
\end{equation}
is the factorized matrix. Therefore, a linear neural network parameterizes a linear function so that such networks are of limited use for general machine learning problems. But understanding training algorithms for such networks is still challenging which is the reason why they are studied intensively in order to understand key principles and, in particular, the implicit regularization of gradient flow/descent for learning them. Compared to linear diagonal networks the training dynamics for learning linear fully connected networks is more complicated due to the noncommutativity of matrix multiplication, and since not only the diagonal, i.e., the singular values will move in general, but also the singular vectors.

\medskip



We start with a rather simple problem, which however, reveals some intuition about the dynamics of learning deep networks. For a fixed symmetric matrix 
$\bW_\star = \bW_\star^T \in \R^{d \times d}$
we consider the square loss
\[
\mathcal{K}(\bW) = \frac{1}{2} \|\bW - \bW_\star\|_2^2. 
\]
Similarly to before, we plug the network \eqref{lin-network}, where all $\bW^{(l)} \in \R^{d \times d}$ have the same dimension, into $\mathcal{K}$ in order to obtain the loss function
\begin{align}\label{eq:matrix_fullrank}
\mathcal{L}_{ln}(\bW^{(1)},\cdots,\bW^{(L)}) &= \mathcal{K}(\bW^{(L)} \cdots \bW^{(1)}) \notag \\
&= 
\|\bW^{(L)}\cdots\bW^{(1)}-\bW_\star\|_F^2,
\end{align}
which can be interpreted as training a fully connected linear network on a square loss, given that the training data spans the input space, see \cite{chou2024gradient}.
We consider gradient flow for minimizing $\mathcal{L}_{ln}$, i.e.,
\begin{equation}\label{grad-flow-E}
\frac{d}{dt} \bW^{(\ell)}(t) = - \nabla_{\bW^{(\ell)}} \mathcal{L}_{ln} (\bW^{(1)}(t),\hdots, \bW^{(L)}(t)),
\end{equation}
as well as gradient descent
\begin{equation}\label{grad-descent-W}
\bW^{(j)}(k+1) = \bW^{(j)} (k) - \eta \nabla_{\bW^{(\ell)}} \mathcal{L}_{ln} (\bW^{(1)}(k),\hdots, \bW^{(L)}(k)).
\end{equation}
We are interested in the dynamics of the product matrix
\begin{equation}\label{prod-matrix-k}
\bW(k) = \bW^{(L)}(k)\cdots\bW^{(1)}(k).
\end{equation}

Note that $\mathcal{K}(\bW)$ has the unique minimizer $\bW = \bW_\star$ so that we formally are not in an overparameterized regime. (The minimizers of $\mathcal{L}_{ln}$ are all tuples $(\bW^{(L)}, \hdots, \bW^{(1)})$ such $\bW^{(L)} \cdots \bW^{(1)} = \bW_\star$.) Nevertheless, as we will see, early stopping of the learning dynamics may have a regularizing effect. One may think of $\bW_\star$ as a noisy version of a low rank ground truth matrix so that regularization towards low rank (see below) may be able to remove noise.

A common type of initialization in theoretical works on linear training set-ups are scaled identity matrices since they facilitate the analysis. However, identically initializing $\bW^{(j)}(0) = \alpha I$, for some $\alpha > 0$ and all $j \in [L]$, induces a spectral cut-off of negative eigenvalues of the product $\widetilde\bW(t) := \bW^{(L)}(t)\cdots\bW^{(1)}(t)$ similar to the orthant restriction in Theorem \ref{thm:dln_ell_1} \cite[Theorem 1.1]{chou2024gradient}. This restriction can be lifted by considering a slightly perturbed identical initialization of the form
\begin{equation}\label{GD:init:W}
\bW^{(1)}(0) = (\alpha- \beta) I, \quad \bW^{(j)}(0) = \alpha I, \quad j = 2,\hdots,L,
\end{equation}
where $0 < \beta < \alpha$ are suitable (small) scalars.
In contrast to identical initialization this allows the dynamics also to reach matrices with negative eigenvalues.

\begin{theorem}[\cite{chou2024gradient}, Theorem 1.3]\label{theorem:PerturbedInitialization_NonAsymptoticConvergence}
Let $L\geq 2$ and let $\Wstar= \Wstar^T \in \R^{d \times d}$ with eigendecomposition $\Wstar = \bV\Lambda \bV^\top \in \mathbb{R}^{d \times d}$, 
where $\Lambda = \diag(\lambda_i : i \in [d])$. Consider the gradient descent iterations $\bW^{(1)}(\tdisc), \hdots, \bW^{(L)}(\tdisc) \in \mathbb{R}^{d\times d}$ in \eqref{grad-descent-W} 
initialized as in \eqref{GD:init:W} with 
parameters satisfying $0 < \frac{\beta}{c-1} < \alpha $ and $c=\max\{c'\in\mathbb{R}:1=(c'-1)(c')^{N-1}\}\in(1,2)$. Let $\bW(\tdisc)$ be the product in \eqref{prod-matrix-k}.
If the step size $\eta$ satisfies
\begin{equation}\label{cond:stepsize:thm}
    0 < \eta < 
    \frac{1}{9L (c\cdot\max(\alpha,\|\Wstar\|^{\frac{1}{L}}))^{2L-2}},
\end{equation}
then $\bWprod(\tdisc)$ converges to $\Wstar$ as $\tdisc \to \infty$. Moreover, the error matrix $E(\tdisc) = \bV^T \bW(\tdisc) \bV - \Lambda$ is diagonal and satisfies
\[
|E_{ii}(\tdisc)| \leq \varepsilon L |\lambda_{i}|^{1-1/N}
\]
for $\tdisc \geq T(\lambda_i,\epsilon,\alpha,\beta,\eta)$ defined in \cite{chou2024gradient} if $|\lambda_i| \geq \alpha^L$.
For a fixed threshold $\epsilon>0$,
$T(\lambda_i,\epsilon,\alpha,\beta,\eta)$
is proportional to 
$\log(\lambda_i)^{-1}$ for $L=2$ and proportional to $\lambda_i^{\frac{2}{L}-2}$ for $L>2$.
\end{theorem}

\begin{figure}
\begin{subfigure}{.45\textwidth}
  \centering
  \includegraphics[width=1\linewidth]{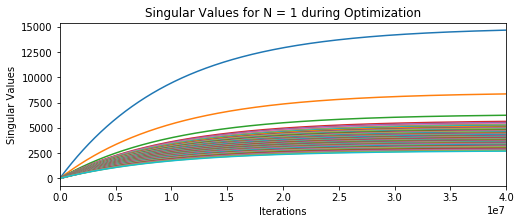}
  \caption{Singular Value Dynamics for $L=1$.}
  \label{fig: MNIST: Noise: SingularValues: N=1}
\end{subfigure}
\begin{subfigure}{.47\textwidth}
  \centering
  \includegraphics[width=1\linewidth]{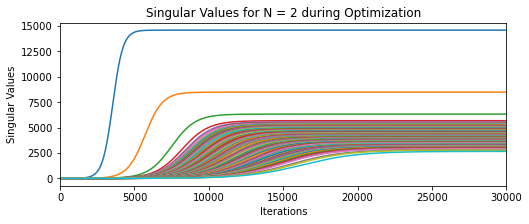}
  \caption{Singular Value Dynamics for $L=2$.}
  \label{fig: MNIST: Noise: SingularValues: N=2}
\end{subfigure}
\begin{subfigure}{.45\textwidth}
  \centering
  \includegraphics[width=1\linewidth]{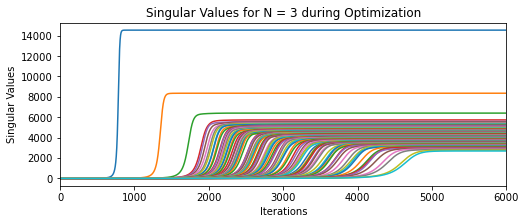}
  \caption{Singular Value Dynamics for $L=3$.}
  \label{fig: MNIST: Noise: SingularValues: N=3}
\end{subfigure}
\hspace{10mm}
\begin{subfigure}{.47\textwidth}
  \centering
  \includegraphics[width=1\linewidth]{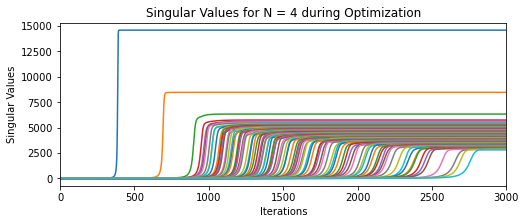}
  \caption{Singular Value Dynamics for $L=4$.}
  \label{fig: MNIST: Noise: SingularValues: N=4}
\end{subfigure}
\caption{Illustration of singular values of $\bWprod(\tdisc)$ during optimization on MNIST. Note that deeper factorization makes the convergence of each eigenvalue sharper, resulting in more distinguishable dynamics between eigenvalues.}
\label{fig: MNIST: Noise: SingularValues}
\end{figure}

Theorem \ref{theorem:PerturbedInitialization_NonAsymptoticConvergence} shows that the convergence speed of the eigenvalues depends on the magnitude of the ground truth eigenvalues. The larger the eigenvalue, the faster it is approximated by the corresponding eigenvalue of $\bW(\tdisc)$. Such a separation of eigenvalues by convergence rate becomes stronger the deeper the network is, as shown in Figure \ref{fig: MNIST: Noise: SingularValues}. An implication of this is that early stopping
of gradient descent leads to a \corr{low-rank} approximation of \corr{$\bW_\star$}, where the rank depends on the stopping time.




Let us also point out a fundamental difference between Theorem~\ref{theorem:PerturbedInitialization_NonAsymptoticConvergence} and the other results such as Theorem~\ref{thm:dln_ell_1} and Theorem~\ref{thm:matrix_gf_min_norm} stated below. The iterates are able to converge to not necessarily positive (definite) solutions. This is due to the fact that the initialization 
in \eqref{GD:init:W} is not identical for all layers.

Theorem~\ref{theorem:PerturbedInitialization_NonAsymptoticConvergence}
also suggests that early stopping of gradient descent on a factorization can be interpreted as a version of implicit bias towards low rank, where we now do not consider the limit of the dynamics but rather a certain time interval where the dynamics stays close to a matrix of a certain rank.
In order to make this statement precise,
we introduce the effective rank of a matrix $\bW$ as
\begin{equation}
    r(\bW) = \frac{\|\bW\|_*}{\|\bW\|_{2 \to 2}},
\end{equation}
where $\|\bW\|_*, \|\bW\|_{2 \to 2}$ are the nuclear and spectral norm of $\bW$, respectively, see Section~\ref{sec:notation}.
Moreover, for simplicity, we return to identical initialization, i.e., $\beta=0$ in \eqref{GD:init:W}, and restrict ourselves to a positive semi-definite ground-truth $\bW_\star$ for which the consequent spectral cut-off has no effect. Both gradient descent and gradient flow are analyzed in \cite[Theorems~3.1 and 3.5]{chou2024gradient}. For simplicity we only state the result for gradient flow here.


\begin{theorem}[\cite{chou2024gradient}, Theorem 3.1] 
\label{thm:dmf_implicit}
Let $\Wstar \in \mathbb{R}^{d\times d}$ be a symmetric ground truth with eigenvalues $\lambda_1 \geq \dots \geq \lambda_n \geq 0$ and denote by $(\Wstar)_s$ the best rank-$s$ approximation of $\Wstar$ (in the Frobenius norm). Consider the gradient flow \eqref{grad-flow-E} with initialization \eqref{GD:init:W}, where $\alpha > 0$ and $\beta=0$. Let $\bW(t)$ be the corresponding product flow \eqref{prod-matrix-k}.
    Set $s' = \max\{s:\lambda_s > \alpha^L \}$. If $L=2$ and $\alpha^2 < \lambda_s$, then
    \begin{equation}\label{eq:EffectiverRankContinuousPSD}
		|r((\Wstar)_s) - r(\bWprod(t))| \leq 2 \epsilon \; r((\Wstar)_s) + C\frac{d-s'}{\|\Wstar\|} \alpha^2
    \end{equation}
    for all $t\in T_{\alpha,\epsilon,\lambda_s}$, where $T_{\alpha,\epsilon,\lambda_s} = [a_{\alpha,\epsilon,\lambda_s},b_{\alpha,\epsilon,\lambda_s}]$ is an interval defined in \cite{chou2024gradient} that may be empty if the spectral gaps between different eigenvalues of $\bW_\star$ are too small.
\end{theorem}


\medskip

{\textbf{Interpretation of dynamics as Riemannian gradient flow.}}
The representation of a matrix as a fully connected linear network, i.e., a product of several matrices, in a loss function and the corresponding gradient flow on the individual matrices, may effectively change the underlying (Riemannian) geometry \cite{Bah2021learning}. Although the connection to implicit regularization is not fully understood, this point of view may give some intuition and may be useful for the further development of the theory. Therefore, we will briefly delve into this topic before continuing the discussion on implicit regularization in matrix recovery problem.

For a general continuously differentiable loss function  $\mathcal{K} : \R^{d_0 \times d_L} \to \R$ on matrices, we consider
the loss on the factorization $\bW = \bW^{(L)} \cdots \bW^{(1)}$ with $\bW^{(\ell)} \in \R^{d_\ell \times d_{\ell-1}}$ defined as
\[
\mathcal{L}(\bW^{(1)}, \hdots, \bW^{(L)}) = \mathcal{K}(\bW^{(L)} \cdots \bW^{(1)}).
\]
We aim at minimizing $\mathcal{L}$ (and thereby $\mathcal{K}$) via the gradient flow
\begin{equation}\label{grad-flow-general}
\frac{d}{dt} \bW^{(\ell)}(t) = - \nabla_{\bW^{(\ell)}} \mathcal{K}(\bW^{(1)}(t),\hdots,\bW^{(L)}(t)), \quad \ell=1,\hdots,L.
\end{equation}
A useful observation is that the terms
\[
(\bW^{(\ell+1)})^\top(t)\bW^{(\ell+1)}(t) - \bW^{(\ell)}(t)(\bW^{(\ell)})^\top(t) = \Delta^{(\ell)}, \quad \ell=1,\hdots, L-1,
\]
are constant in time, i.e., they are invariants of the ODE system \eqref{grad-flow-general} \cite{Bah2021learning,arcoha18,arora2018convergence,chlico23}. In particular, for so-called balanced initialization
\begin{equation}\label{balancedness}
(\bW^{(\ell+1)}(0))^\top\bW^{(\ell+1)}(0) = \bW^{(\ell)}(0)(\bW^{(\ell)}(0))^\top, \quad \ell=1,\hdots,L-1
\end{equation}
it holds $(\bW^{(\ell+1)}(t))^\top \bW^{(\ell+1)}(t) = \bW^{(\ell)}(t)(\bW^{(\ell)}(t))^\top$ for all $t \geq 0$.
In this case, one can show \cite{arcoha18,Bah2021learning} that the product
\begin{equation}\label{prod-flow}
\bW(t) = \bW^{(L)}(t) \cdots \bW^{(1)}(t)
\end{equation}
satisfies the equation
\begin{equation}\label{flow-product}
\frac{d}{dt} \bW(t) = - \mathcal{B}_{\bW(t)}(\nabla \mathcal{K}(\bW(t))).
\end{equation}
where the operator $\mathcal{B}_{\bW}$ is defined as
\[
\mathcal{B}_{\bW}(\bZ) = \sum_{j=1}^N (\bW \bW^T)^{\frac{N-j}{N}} \bZ  (\bW^T \bW)^{\frac{j-1}{N}}.
\]
It is shown in \cite{Bah2021learning} that if $\bW$ is of rank $k$, then the restriction of $\mathcal{B}_{\bW}$
to the tangent space of the manifold $\mathcal{M}_k$ of rank-$k$-matrices in $\R^{d_0 \times d_L}$ is self-adjoint and positive definite, hence invertible. The inverse $\overline{\mathcal{B}}_{\bW}^{-1}$ of this restriction defines a Riemannian metric on $\mathcal{M}_k$
via
\[
g_{\bW}(\bZ_1,\bZ_2) = \langle \overline{\mathcal{B}}_{\bW}^{-1}(\bZ_1),\bZ_2\rangle_F
\]
for all $\bZ_1,\bZ_2$ in the tangent space of $\mathcal{M}_k$ at $\bW$. The corresponding Riemannian gradient can be expressed
\[
\nabla^g \mathcal{K}(\bW) = \mathcal{B}_{\bW}(\nabla \mathcal{K}(\bW)),
\]
so that the flow equation \eqref{flow-product} for the product can be written as the Riemannian gradient flow
\[
\frac{d}{dt} \bW(t) = - 
\nabla^g \mathcal{K}(\bW(t)).
\]
In this sense, the reparameterization of matrix as a product of matrices may change the underlying geometry.

The article \cite{comeve23} contains more information on this Riemannian metric and numerical experiments concerning its influence on implicit regularization. These suggests that Riemannian gradient flow on the factorization favors minimizers where the Riemannian volume is maximized. However, a clear understanding of this phenomenon is still missing.

In \cite{Bah2021learning}, the gradient flow for the particular loss functions
\begin{align*}
\mathcal{K}(\bW) &= \frac{1}{2}\|\bW \bX - \bY\|_F^2, \\
\mathcal{L}(\bW^{(1)},\hdots,\bW^{(L)}) & = \mathcal{K}(\bW^{(L)}\cdots \bW^{(1)})
\end{align*}
is studied, where $\bX \in \R^{m \times d_0}$ and $\bY \in \R^{m \times d_L}$ are matrices of input and output data.
It is basically shown that the product flow \eqref{prod-flow} related to the gradient flow \eqref{grad-flow-general} converges to a global minimizer of $\mathcal{K}$ restricted to a manifold of rank-$k$ matrices for a certain $k$ (for almost all balanced and unbalanced initializations), see \cite{Bah2021learning} for details. This is remarkable because $\mathcal{L}$ is nonconvex.


\medskip

{\textbf{Implicit bias in underdetermined matrix recovery problems.}} We now consider implicit regularization in the context of low rank matrix recovery, see also Section \ref{sec:LowRank}. For a general linear map $\mathcal{A} : \R^{d \to d} \to \R^m$ with $m < d^2$ which can be described in the form \eqref{matrix-meas} with measurement matrices $\bA_1,\hdots,\bA_m \in \R^{d \times d}$, we consider the loss function $\mathcal{K} : \R^{d \times d} \to \R$ defined as
\[
\mathcal{K}(\bW) = \frac{1}{2} \|\mathcal{A}(\bW)- \by\|_2^2
\]
and plug in a factorization $\bW = \bW^{(L)} \cdots \bW^{(1)}$ with $\bW^{(\ell)} \in \R^{d\times d}$ (for simplicity we consider only square matrices $\bW$ and $\bW^{(\ell)}$) resulting in
\begin{equation}\label{L-factorized}
\mathcal{L}(\bW^{(1)}, \hdots, \bW^{(L)}) = \mathcal{K}(\bW^{(L)}\cdots \bW^{(1)}) = \frac{1}{2} \|\mathcal{A}(\bW^{(L)}\cdots \bW^{(1)}) - \by\|_2^2.
\end{equation}
We consider the corresponding gradient flow \eqref{grad-flow-general} with the associated product flow $\bW(t)$ in \eqref{prod-flow}. We choose identical initialization, i.e,
\begin{equation}\label{init-alpha}
\bW^{(\ell)}(0) = \alpha I, \quad \ell= 1,\hdots, L,
\end{equation}
for some $\alpha > 0$ and $I$ being the $d \times d$ identity matrix, resulting in $\bW(0) = \alpha^L I$.
Note that this initialization is in particular balanced, i.e., satisfies \eqref{balancedness}.

Taking intuition from the vector case discussed in the previous section, we expect that, for $L\geq 2$ and small initialization scale $\alpha$, gradient flow exhibits an implicit bias towards low rank in the limit as $t \to \infty$. Numerical experiments confirm this intuition. In Figure~\ref{Fig:LowRankRecovery}, the success probability for recovering a matrix in $\R^{20 \times 20}$ of rank $2$ via gradient flow applied to linear networks of depth $L=2,3,4$ and small initialization $\alpha$. The $x$-axis corresponds to the used number of random measurements.
While the case $L=2$ requires comparably many measurements (but still allows \corr{one} to solve an underdetermined system), the cases \corr{$L=3$} and $L=4$ reach success at almost the optimal number of measurements. In fact, the number of degrees of freedom to describe a $d \times d$ matrix of rank $r$ is $r(2d-r)$\corr{,} which for $d=20$ and $r=2$ gives $76$.

\begin{figure}
{\centering{
\includegraphics[width=0.9\textwidth]{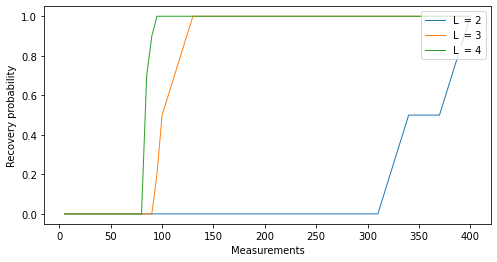}
\caption{Success probability for the recovery of $20 \times 20$ matrices of rank $2$ from Gaussian measurements with gradient flow on deep linear networks of depth $L=2,3,4$.}
\label{Fig:LowRankRecovery}}}
\end{figure}

The theory cannot yet fully explain this observation, but partial results are available.
The following result for the special case that the measurement matrices $\bA_1,\hdots,\bA_m$ in \eqref{matrix-meas} commute has been shown in 
\cite[Theorem 2, Proposition 2]{arora2019implicit}.
 A similar result for symmetric factorizations in the two-layer case
 $\bW = \bW^{(1)} (\bW^{(1)})^T$ 
 has been derived in \cite[Theorem 1]{gunasekar2017implicit}. 




\begin{theorem}[\cite{arora2019implicit}]\label{thm:matrix_gf_min_norm}
Let $L \geq 2$ and suppose that the matrices $A_1,\ldots,A_m$ commute, i.e., $\bA_i \bA_j = \bA_j \bA_i$ for all $i,j = 1,\hdots,m$. Consider the product flow $\bW(t)=\bW_{\alpha}(t)$ in \eqref{flow-product} 
$\mathcal{L}$ in \eqref{L-factorized} with initialization \eqref{init-alpha} with $\alpha > 0$. Assume that $\bW_{\infty,0}:=\lim_{\alpha\to 0} \lim_{t \to \infty} \bW_{\alpha}(t)$ exists and satisfies $\mathcal{K}(\bW_{\infty,0}) = 0$.
 Then it holds that
    \begin{equation}\label{eq:min_nuc}
        \bW_{\infty,0}\in\argmin_{\bW\succeq 0,\mathcal{K}(\bW)=0}\|\bW\|_*.
    \end{equation}
\end{theorem}

The proof is based on the fact that due to the commutativity of the $\bA_j$, the eigenspaces of $\bW(t)$ remain constant with respect to $t$ so that it is enough to consider the dynamics of the diagonal matrices containing the eigenvalues. This essentially reduces the problem  to one of analyzing linear diagonal networks discussed in the previous section. 
Furthermore, Theorem \ref{thm:matrix_gf_min_norm} explicitly assumes convergence of gradient flow, whereas the previous results proved convergence in addition to characterizing implicit regularization.

Note that the assumption that the matrices $\bA_1,\hdots,\bA_m$ commute is quite strong. In particular, one can only find at most $d$ commuting matrices in $\R^{d \times d}$ that are linearly independent. However, counting degrees of freedom it is not possible to have unique recovery of even a rank one matrix $\bX \in \R^{d \times d}$ from fewer than $2d-1$ matrices. Therefore, the assumption of commuting measurements is not realistic for most settings of practical relevance.

For the case of $L=2$ layers, this assumption could be removed in \cite{sost25,stoger2021small,jin2023incremental}, where it is instead assumed that the measurement map $\mathcal{A}$ satisfies the matrix version of the RIP \eqref{matrix-RIP}. 
The authors use random initialization $\bW^{(1)} = \alpha \bU$, $\bW^{(2)} = \alpha \bV$ for independent standard Gaussian random matrices $\bU, \bV \in \R^{d \times d}$. For the ground truth matrix $\bX \in \R^{d \times d}$ of rank $r$ they introduce the condition number $\kappa = \|\bX\|_{2 \to 2}/\sigma_r$, where $\sigma_r$ denotes the $r$-th singular value of $\bX$. They assume that the measurement map $\mathcal{A}$ satisfies the RIP \eqref{matrix-RIP} with constant
\begin{equation}\label{cond-deltar}
\delta_{2r+1} \leq \frac{c}{\kappa^3 \sqrt{r}}
\end{equation}
and that $\by = \mathcal{A}(\bX)$.
They show, for small enough step size $\eta$ and small enough initialization
$\alpha$, that after a certain number $K$ of iterations of gradient descent,
\[
\bW^{(\ell)}(k+1) = \bW^{(\ell)}(k) - \eta \nabla_{\bW^{(\ell)}}\mathcal{L}(\bW^{(1)}(k),\bW^{(2)}(k)), \quad \ell=1,2, \quad k=1,2,\hdots,
\]
the product $\bW(K) = \bW^{(2)}(K)\bW^{(1)}(K)$ of the iterates  satisfies, with high probability,
\[
\|\bW(K) - \bX\|_{2 \to 2} \leq \alpha^{3/5}\|\bX\|_{2 \to 2}^{7/10}.
\]
We refer to \cite[Theorem 3.3]{sost25} for details and a slightly more general version. Note that
\corr{plugging} the condition \eqref{cond-deltar} into the RIP bound
\eqref{cond:m:RIP-matrix} for Gaussian measurement maps leads to the condition
\[
m \geq C \kappa^6 r^{2} d.
\]
It is currently not clear whether the dependence on the condition number $\kappa$ of $\bX$ and the quadratic scaling in the rank $r$ is needed.

Extensions to deeper networks with $L\geq 3$ layers are currently open.


\subsection{Early alignment for two-layer ReLU networks}

Compared to the simplified linear settings discussed in the previous chapters, the analysis of implicit regularization for gradient flow/descent for learning neural networks becomes significantly more challenging when the activation $\sigma$ in \eqref{eq:FNN} is nonlinear. Some results are available at least for $L=2$ layers,
see e.g.\ \cite{boursier2022gradient,boursier2024,chistikov2023learning,chizat2020implicit,maennel2018quantizes,min2024early}. 
Two-layer neural networks $h : \R^d \to \R$ can be written in the form
\[
h_{\btheta}(\bx) = h_{\ba,\bW}(\bx) = \sum_{j=1}^d a_j \sigma(\bw_j^T \bx), \quad \bx \in \R^d, 
\]
where $a_1,\hdots,a_d \in \R$ and $\bw_1,\hdots,\bw_d \in \R^d$ (corresponding to the rows of the weight matrix $\bW^{(1)}$ in \eqref{eq:FNN}). The parameters are assembled in the vector $\btheta = (a_j,\bw_j)_{j = 1}^d$. Given training data $(\bx_i,\by_i) \in \R^{d} \times \R$, $i=1,\hdots,n$, we consider learning $h_{\btheta}$ by minimizing the empirical risk \eqref{emp:risk}, i.e.,
\[
\mathcal{L}(\btheta) = \frac{1}{n}\sum_{i=1}^n \ell(h_{\btheta}(\bx_i),y_i)
\]
via gradient flow \eqref{def:grad-flow} or gradient descent \eqref{def:grad-descent}.

The ReLU activation function $\sigma(z) = \max\{0,z\}$ has been given particular attention in several works \cite{boursier2022gradient,boursier2024,maennel2018quantizes} due to its popularity in applications. (Note that the gradient flow has to be replaced by a subgradient flow due to the non-differentiability of $\sigma$ in $0$.)

For small initialization, again implicit regularization can be observed in the sense that in the first phase of the training, many neurons align with a few directions \cite{boursier2024,maennel2018quantizes}. This phenomenon is called early alignment. 
Figure~\ref{Fig:EarlyAlign} illustrates this for the case $d=2$ by displaying all the inner weight vectors $\bw_j(k) \in \R^2$ at iteration $k=3000$. The neurons point in only $3$ possible directions.

\begin{figure}
{\centering{
\includegraphics[width=0.7\textwidth]{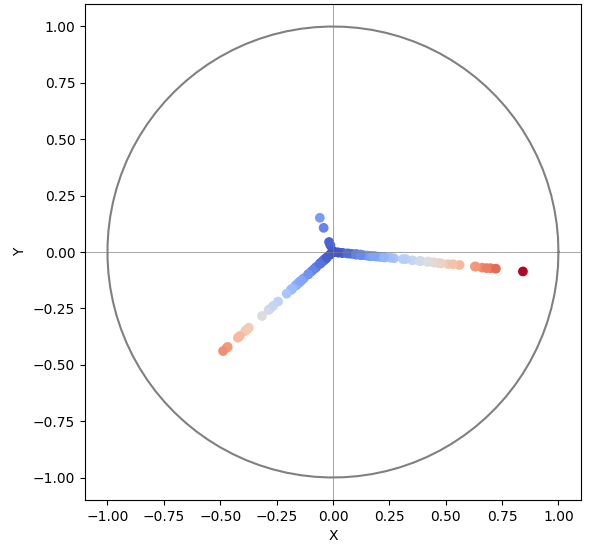}
\caption{Early alignment of inner neurons of two-layer ReLU network after $3000$ iteration of gradient descent with small initialization.}
\label{Fig:EarlyAlign}}}
\end{figure}

The key idea of the analysis is that the flow can be approximated by the one, where the (small) outputs of the neural network $h_{\btheta(t)}(\bx_i)$ can be replaced by $0$. In this situation, the dynamics for the single neurons decouples and can be described by
\[
\frac{d \bw_j(t)}{dt} \in a_j(t) \mathcal{D}(\bw_j(t)),  \qquad \frac{d a_j(t)}{dt} \in \mathcal{D}(\bw_j(t))^T \bw_j(t), 
\]
where
\begin{align*}
\mathcal{D}(\bw) & = \left\{
-\frac{1}{m} \sum_{j=1}^m \eta_j \bx_j \, \partial_1 \ell(0, \by_j)^T 
\;\middle|\;
\eta_j =\left\{
\begin{array}{ll}
1 & \langle \bw, \bx_j \rangle > 0 \\
{[0,1]} & \langle \bw, \bx_j \rangle = 0 \\
0 & \langle \bw, \bx \rangle < 0
\end{array}\right\}
\right\}.
\end{align*}
Define $D(\bw) = \argmin_{D \in \mathcal{D}(\bw)} \|D\|_F$. This dynamics of the directions of the neurons $\bw_j(t)/\|\bw_j(t)\|_2$ can approximately be described by a gradient ascent flow for the function
\[
G : S^{d-1} \to \R, \quad G(\bw) = \langle D(\bw), \bw \rangle.
\]
Using these observations, it is basically shown in \cite{boursier2025early,maennel2018quantizes} that the directions $\frac{\bw_{j}(t)}{\|\bw_j(t)\|_2}$ of the neurons $\bw_{j}(t)$
approach critical points of the function $G$ above\corr{;} see \cite{boursier2025early,maennel2018quantizes} for the \corr{detailed technical} statements.
As a consequence, if the vectors $\bv_{1},\hdots,\bv_{s} \in S^{d-1}$ are the critical points of $G$, then the neural network $h_{\btheta(t)}$ approaches a function of the form
\[
g(\bx) = \sum_{j=1}^s b_j \sigma(\bv_j^T \bx)
\]
in the first phase of the training. If the number $s$ of critical points of $G$ is small -- which often seems to be the case -- this shows that also in this situation gradient flow induces an implicit bias towards sparse solutions in this sense. (Note that this is different notion of sparsity than observed for diagonal linear networks.)

We remark that the approximation of the dynamics by the one where the outputs of the neural network $h_{\btheta}(\bx_i)$ are replaced by $0$ is only valid at the initial phase when these outputs are small. When they become large, it seems hard to predict the dynamics in general and in particular, to say something about the limit. Nevertheless, the dynamics often proceeds in phases with almost constant output empirically, and a similar analysis seems possible also for these almost stagnant phases.

Generalizations of these findings to (ReLU) networks with more than two layers seem to be widely open.


\subsection{Discussion and Future Research Directions}

Understanding the implicit regularization of gradient descent and its variations is key to understanding generalization in deep learning in the overparameterized scenario. As the previous sections illustrated, implicit regularization of neural network learning seems to be towards solutions of lower complexity than the mere number of parameters used for representing the neural network -- at least for suitable (small) initialization.
Following Occam's razor principle that among several models nature chooses the simplest one, the ground truth models describing reality should be of rather low complexity. If this is true for a particular application, then implicit regularization towards low complexity solutions gives an explanation of why learned neural networks can lead to good generalization despite overparameterization.

It is a major research challenge to understand this phenomenon in more detail and for more realistic deep neural network models as well as for commonly used variants of gradient descent such as preconditioned methods like Adam. While in the examples above, we have seen that sparsity and low rank structures can arise, it may be that the low complexity structures arising in deep neural network learning are more complicated and it is also a matter of finding out the nature of such low complexity models, i.e., ways of representing large scale neural networks with much fewer parameters.

We end this article with a general question. Observing that the implicit regularization in training largely overparameterized neural networks may lead to networks of much lower complexity, one may question whether there are ways to learn such low complexity models more efficiently with few parameters from the beginning without the need of first taking a detour to networks with huge number of parameters. This is somewhat similar to the original idea of compressive sensing that it suffices to take much fewer measurements than previously believed necessary in order to efficiently recover a sparse vector in high dimensions. However, in contrast to compressive sensing, the deep learning problem is highly nonlinear, and at this point it is not  
not clear whether passing to higher dimension is necessary for computational reasons.
It seems very interesting but also challenging to attack such research questions.

\newpage
\bibliography{bibliography}
\bibliographystyle{plain}

%
%







\end{document}

%% file: Shortcuts.tex
\newcommand{\ba}{\boldsymbol{a}}
\newcommand{\bA}{\boldsymbol{A}}
\newcommand{\bb}{\boldsymbol{b}}
\newcommand{\bB}{\boldsymbol{B}}

\newcommand{\bI}{\boldsymbol{I}}

\newcommand{\bt}{\boldsymbol{t}}

\newcommand{\bu}{\boldsymbol{u}}
\newcommand{\bU}{\boldsymbol{U}}
\newcommand{\bv}{\boldsymbol{v}}
\newcommand{\bV}{\boldsymbol{V}}
\newcommand{\bw}{\boldsymbol{w}}
\newcommand{\bW}{\boldsymbol{W}}
\newcommand{\bx}{\boldsymbol{x}}
\newcommand{\bX}{\boldsymbol{X}}
\newcommand{\by}{\boldsymbol{y}}
\newcommand{\bY}{\boldsymbol{Y}}
\newcommand{\bz}{\boldsymbol{z}}
\newcommand{\bZ}{\boldsymbol{Z}}
\newcommand{\Eta}{\boldsymbol{\eta}}

\newcommand{\btheta}{{\boldsymbol{\theta}}}

\newcommand{\bPhi}{\boldsymbol{\Phi}}

\newcommand{\0}{\boldsymbol{0}}
\newcommand{\1}{\boldsymbol{1}}

\newcommand{\dout}{d_{\text{out}}}

\newcommand{\calA}{\mathcal{A}}

\newcommand{\calD}{\mathcal{D}}

\newcommand{\calH}{\mathcal{H}}

\newcommand{\calR}{\mathcal{R}}

\newcommand{\calN}{\mathcal{N}}

\DeclareMathOperator*{\argmin}{arg\,min}

\renewcommand{\subset}{\subseteq}
\renewcommand{\hat}{\widehat}
\renewcommand{\tilde}{\widetilde}
\renewcommand{\epsilon}{\varepsilon}

\def\supp{\mathrm{supp}}

\def\<{\big\langle}
\def\>{\big\rangle}
\def\({\Big(}
\def\){\Big)}
\def\calA{\mathcal{A}}

\def\C{\mathbb{C}}
\def\calD{\mathcal{D}}
\def\E{\mathbb{E}}

\def\calH{\mathcal{H}}

\def\N{\mathbb{N}}
\def\calN{\mathcal{N}}

\def\R{\mathbb{R}}
\def\calR{\mathcal{R}}

\def\calX{\mathcal{X}}
\def\calY{\mathcal{Y}}

\newcommand{\St}{\mathbb{S}}

\def\Lls{\mathcal{L}_\text{ls}}

\def\Lgl{\mathcal{L}_\text{glm}}
\def\Lwn{\mathcal{L}_\text{wn}}

\def\xprod{\tilde{\bf x}}

\def\xprodinfty{\xprod_\infty}

\def\sigmainn{\rho}
\def\sigmaout{\sigma}
\def\Wstar{\bW_\star}
\def\tcont{{t}}
\def\tdisc{{k}}
\def\bWprod{\tilde{\bW}}

\def\lrr{{\tilde{\eta}}}

\newcommand{\diag}{\mathrm{diag}}

\newcommand{\pnorm}[2]{\left\| #1 \right\|_{#2}}